\documentclass[a4paper,USenglish,cleveref,thm-restate,numberwithinsect]{lipics-v2021}
\pdfoutput=1
\hideLIPIcs
\nolinenumbers

\usepackage{complexity} % complexity classes
\usepackage{hyphenat} % hyphens for compound words
\usepackage{cite}
\usepackage[export]{adjustbox}

\usepackage{thmtools}

\title{Local problems in trees across a wide range of distributed models}
\author{Anubhav Dhar}
    {Aalto University, Espoo, Finland \and Indian Institute of Technology Kharagpur, India}
    {anubhavdhar@kgpian.iitkgp.ac.in}
    {https://orcid.org/0009-0006-5922-8300}{}
\author{Eli Kujawa}
    {Aalto University, Espoo, Finland \and University of Illinois Urbana-Champaign, USA}
    {ekujawa2@illinois.edu}
    {https://orcid.org/0009-0009-4183-7126}{}
\author{Henrik Lievonen}
    {Aalto University, Espoo, Finland \and \url{https://henriklievonen.fi/}}
    {henrik.lievonen@aalto.fi}
    {https://orcid.org/0000-0002-1136-522X}{}
\author{Augusto Modanese}
    {Aalto University, Espoo, Finland \and \url{https://augusto.modanese.net/}}
    {augusto.modanese@aalto.fi}
    {https://orcid.org/0000-0003-0518-8754}{}
\author{Mikail Muftuoglu}
    {Aalto University, Espoo, Finland}
    {mikail.muftuoglu@aalto.fi}
    {https://orcid.org/0009-0008-7350-5372}{}
\author{Jan Studený}
    {Aalto University, Espoo, Finland}
    {jan.studeny@aalto.fi}
    {https://orcid.org/0000-0002-9887-5192}{}
\author{Jukka Suomela}
    {Aalto University, Espoo, Finland \and \url{https://jukkasuomela.fi/}}
    {jukka.suomela@aalto.fi}
    {https://orcid.org/0000-0001-6117-8089}{}
\authorrunning{A.\ Dhar et al.}
\Copyright{Anubhav Dhar, Eli Kujawa, Henrik Lievonen, Augusto Modanese, Mikail Muftuoglu, Jan Studený, Jukka Suomela}
\keywords{Distributed algorithms, quantum-LOCAL model, randomized online-LOCAL model, locally checkable labeling problems, trees}
\ccsdesc[500]{Computing methodologies~Distributed algorithms}
\funding{This work was supported in part by the Research Council of Finland, Grants 333837 and 359104, the Aalto Science Institute (AScI), and the Helsinki Institute for Information Technology (HIIT).}

\let\H\relax

\usepackage{ammacros}
\usepackage{local-models}

\newcommand*{\Abuc}{\mathcal A^+(b,u,C)}
\newcommand*{\Afails}{\mathcal F}
\newcommand*{\CC}{\mathcal{C}}
\newcommand*{\Sigmabar}{\overline{\Sigma}}
\newcommand*{\Sigmacert}{\Sigma_{\mathcal T}}
\newcommand*{\VV}{\mathcal{V}}
\newcommand*{\algo}{\mathcal{A}}
\newcommand*{\allpaths}{\mathcal{P}}
\newcommand*{\deltree}{\mathcal{T}}
\newcommand*{\iin}{\mathrm{in}}
\newcommand*{\labeling}{\mathcal{L}}
\newcommand*{\lpump}{l_{\mathsf{pump}}}
\newcommand*{\ogruntime}{t_{\mathsf{orig}}}
\newcommand*{\oout}{\mathrm{out}}
\newcommand*{\pathset}{\mathcal{Q}}
\newcommand*{\sC}{\mathsf{C}}
\newcommand*{\sR}{\mathsf{R}}
\newcommand*{\seq}{\mathcal{S}}
\newcommand*{\sigmabaropt}{\overline{\sigma*}}
\newcommand*{\sigmabar}{\overline{\sigma}}
\newcommand*{\skeltree}{\tree_{\mathsf{skel}}}
\newcommand*{\tree}{T}

\renewcommand*{\PP}{\mathcal{P}}

\newcommand*{\SigmaC}[1]{{\Sigma}_{#1}^\mathsf{C}}
\newcommand*{\SigmaR}[1]{{\Sigma}_{#1}^\mathsf{R}}
\newcommand*{\rest}[2]{{#1\upharpoonright_{#2}}}
\newcommand*{\textmultich}[2]{%
   \bigl(\!\bigl(\!\begin{smallmatrix}#1\\#2\end{smallmatrix}\!\bigr)\!\bigr){}}
\newcommand{\displaymultich}[2]{\left(\!{\binom{#1}{#2}}\!\right)}
\newcommand*{\multich}[2]{\mathchoice{\displaymultich{#1}{#2}}
    {\textmultich{#1}{#2}}
    {\textmultich{#1}{#2}}
    {\textmultich{#1}{#2}}}

\DeclareMathOperator{\Class}{\mathsf{Class}}
\DeclareMathOperator{\Type}{\mathsf{Type}}
\DeclareMathOperator{\flexSCC}{\mathsf{flexible-SCC}}
\DeclareMathOperator{\replace}{\mathsf{Replace}}
\DeclareMathOperator{\trim}{\mathsf{trim}}

\DeclareMathOperator*{\argmin}{arg\,min}

\usepackage{csquotes}

\begin{document}

\maketitle

\begin{abstract}
    The \emph{randomized online-LOCAL} model captures a number of models of computing; it is at least as strong as all of these models:
    \begin{itemize}
        \item the classical LOCAL model of distributed graph algorithms,
        \item the quantum version of the LOCAL model,
        \item finitely dependent distributions [e.g.\ Holroyd 2016],
        \item any model that does not violate physical causality [Gavoille, Kosowski, Markiewicz, DISC 2009],
        \item the SLOCAL model [Ghaffari, Kuhn, Maus, STOC 2017], and
        \item the dynamic-LOCAL and online-LOCAL models [Akbari et al., ICALP 2023].
    \end{itemize}
    In general, the online-LOCAL model can be much stronger than the LOCAL model. For example, there are \emph{locally checkable labeling problems} (LCLs) that can be solved with logarithmic locality in the online-LOCAL model but that require polynomial locality in the LOCAL model.

    However, in this work we show that in \emph{trees}, many classes of LCL problems have the same locality in deterministic LOCAL and randomized online-LOCAL (and as a corollary across all the above-mentioned models). In particular, these classes of problems do not admit any distributed quantum advantage.

    We present a near-complete classification for the case of \emph{rooted regular trees}. We also fully classify the super-logarithmic region in \emph{unrooted regular trees}. Finally, we show that in general trees (rooted or unrooted, possibly irregular, possibly with input labels) problems that are global in deterministic LOCAL remain global also in the randomized online-LOCAL model.
\end{abstract}

\newpage

\section{Introduction}

\begin{figure}
    \centering
    \includegraphics[page=1,scale=0.99]{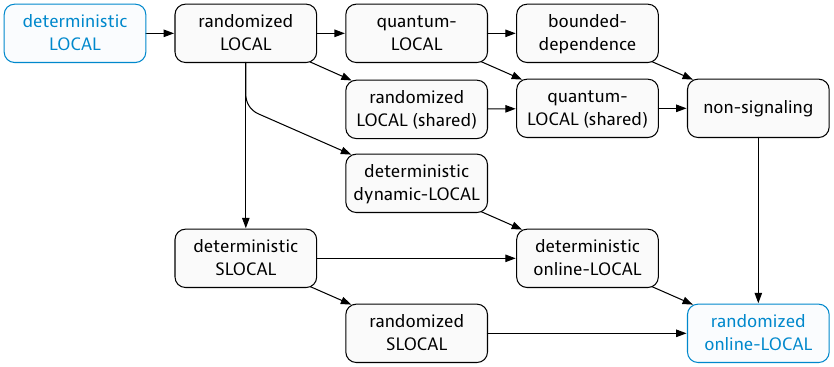}
    \caption{Landscape of models, based on \cite{akbari24_online_arxiv}. In this work we show that for many families of LCL problems, the two extreme models---\emph{\detlcl} and \emph{\rolcl}---are equally strong, and hence the same holds for \emph{all} intermediate models in this diagram.}
    \label{fig:models}
\end{figure}

The \emph{\rolcl model} was recently introduced in \cite{akbari24_online_arxiv};
this is a model of computing that is at least as strong as many other models
that have been widely studied in the theory of distributed computing, as well as
a number of emerging models; see \cref{fig:models}. In particular, different
variants of the \emph{\qlocal} and \emph{\slocal} models are sandwiched between
the classical \emph{\detlcl} model and the \rolcl model.

While the \rolcl model is in general much stronger than the \detlcl model, in this work we show that
in \emph{trees}, for many families of \emph{local} problems these two models
(and hence all models in between) are asymptotically equally strong.
\Cref{fig:landscape-lcls-after} summarizes the relations that we have thanks to
this work; for comparison, see \cref{fig:landscape-lcls-before} for the state of
the art before this work.

\begin{figure}
    \centering
    \includegraphics[width=\textwidth]{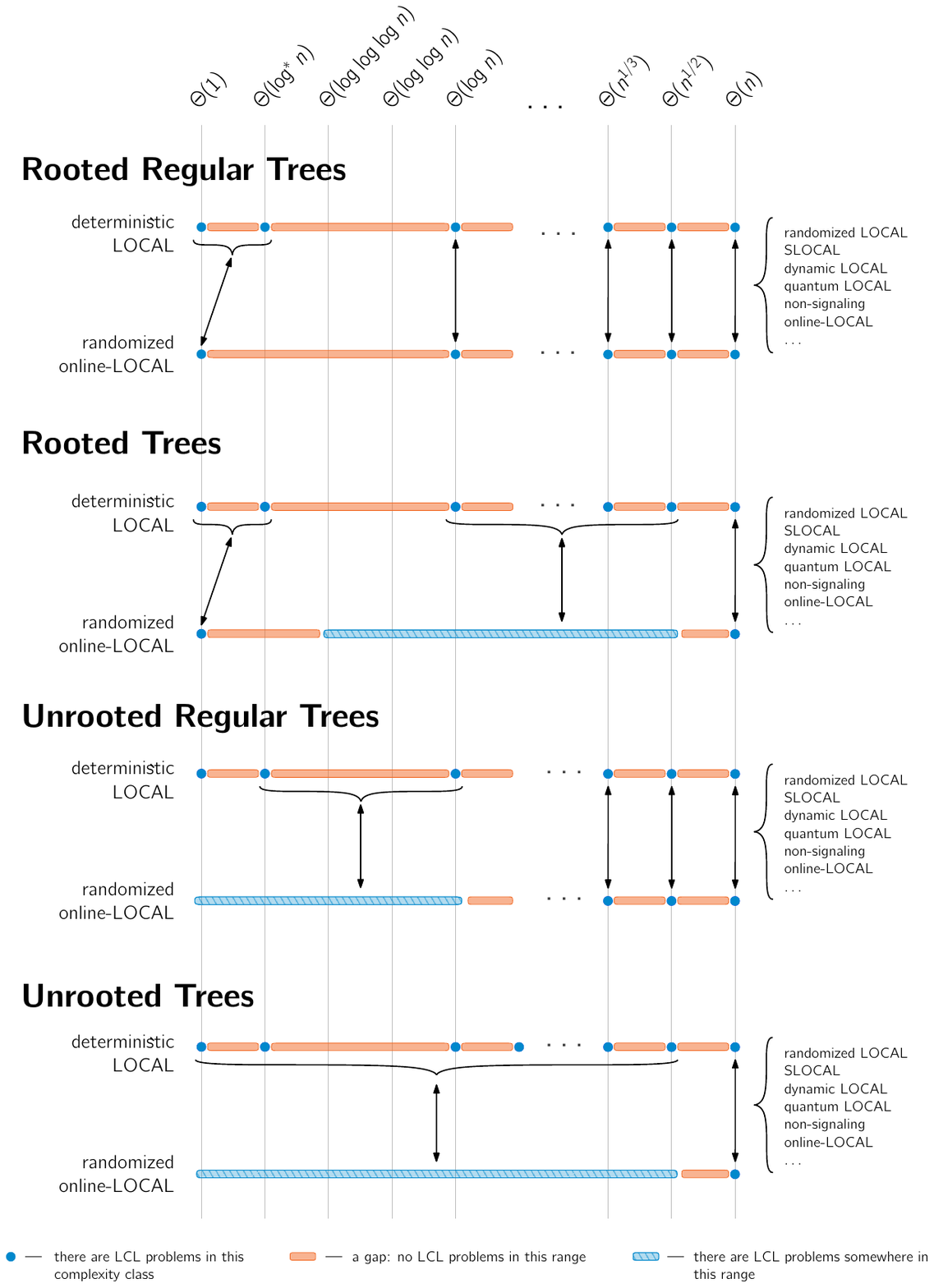}
    \caption{Landscape of LCL problems in trees after this work---compare with \cref{fig:landscape-lcls-before} that shows the state of the art before this work.}
    \label{fig:landscape-lcls-after}
\end{figure}

\begin{figure}
    \centering
    \includegraphics[width=\textwidth]{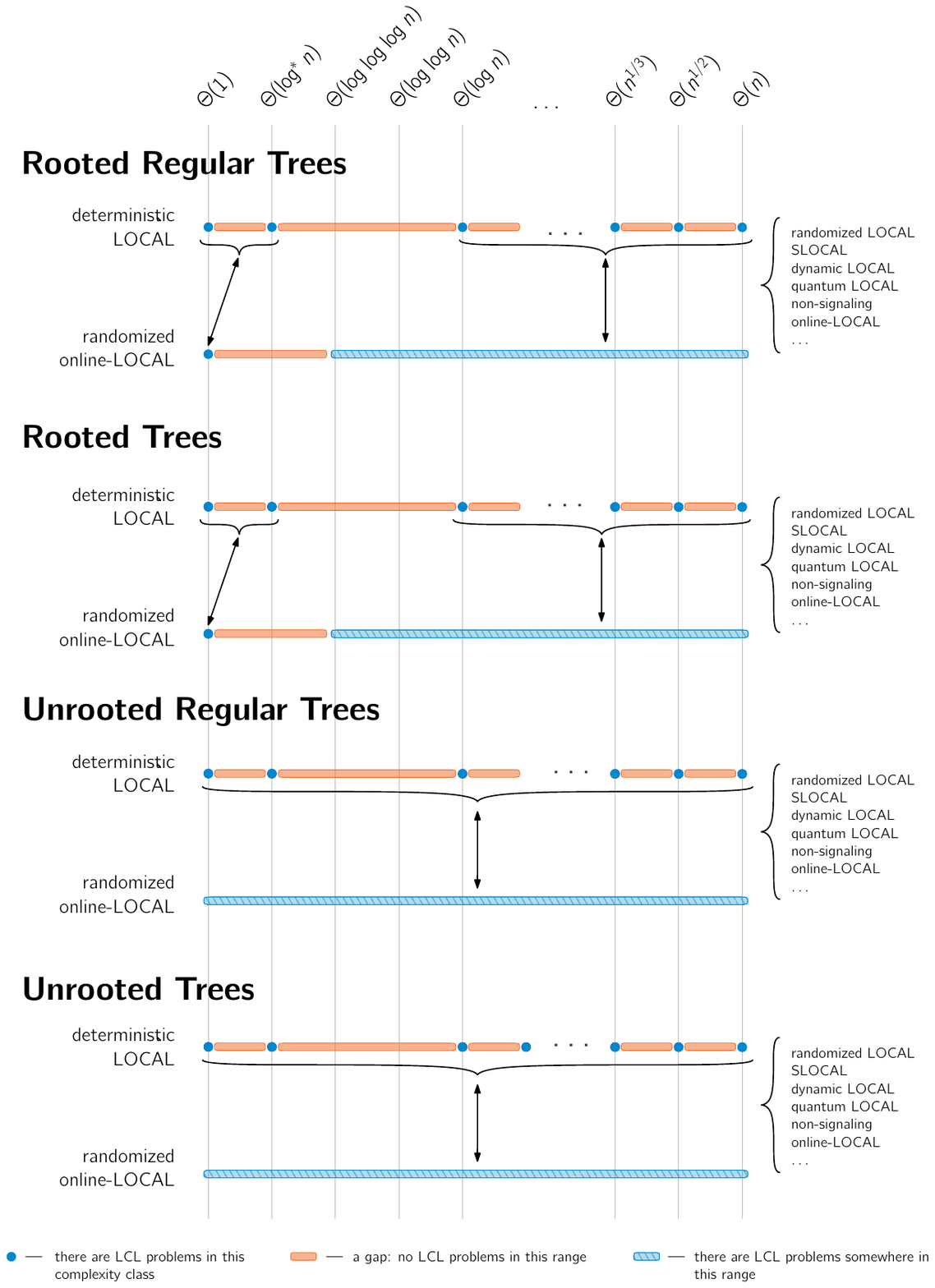}
    \caption{Landscape of LCL problems in trees before this work---compare with \cref{fig:landscape-lcls-after} to see the impact of our new contributions.}
    \label{fig:landscape-lcls-before}
\end{figure}

\subsection{Models}

We will define all relevant models formally in \cref{sec:background}, but for now the following brief definitions suffice:
\begin{itemize}
    \item In the \textbf{\boldmath\detlcl} model, an algorithm $A$ with locality $T$
    works as follows:
    The adversary chooses a graph $G$ and an assignment of polynomially-sized unique identifiers.
    Algorithm $A$ is applied to all nodes \emph{simultaneously in parallel}.
    When we apply $A$ to node $v$, algorithm $A$ gets to see the radius-$T$
    neighborhood of $v$ and, using this information, it has to choose the output
    label of node $v$.
    \item In the \textbf{\boldmath\rolcl} model, an algorithm $A$ with locality $T$ works
    as follows:
    The adversary chooses a graph $G$ and a processing order $\sigma$.
    Then the adversary presents nodes \emph{sequentially} following the order $\sigma$.
    Whenever a node $v$ is presented, algorithm $A$ gets to see the radius-$T$
    neighborhood of $v$ and, using this information \emph{as well as} all
    information it has seen previously and a global source of random bits, it
    has to choose the output label of node $v$.
\end{itemize}
This means that \rolcl is stronger than \detlcl in at least three different
ways: (1)~we have access to shared randomness, (2)~the sequential processing
order can be used to break symmetry, and (3)~there is global memory thanks to
which we can remember everything we have seen so far.
A reader familiar with the \slocal model can interpret it as randomized \slocal
augmented with global memory. 
Note that the adversary is oblivious; it cannot adapt $G$ and $\sigma$ based on
the actions of~$A$.

In general, \onlinelocal (with or without randomness) is much stronger than the
classical \local model (with or without randomness). 
For example, leader election is trivial in \onlinelocal, even with locality $T =
0$ (the first node that the adversary presents is marked as the leader). 
However, what is much more interesting is whether \onlinelocal has advantage
over \local for problems defined using local constraints, such as graph
coloring.

\subsection{Prior Work on LCL Problems}

We will study here locally checkable labeling problems (LCLs) \cite{whatcanbecomputedlocally}; these are problems that can be specified by giving a finite set of valid local neighborhoods. By prior work, we know a number of LCL problems that can separate the models in \cref{fig:models}, for example:
\begin{itemize}
    \item Sinkless orientation has locality $\Theta(\log n)$ in \detlcl,
    locality $\Theta(\log \log n)$ in \randlcl, and locality $\Theta(\log \log
    \log n)$ in \rslocal \cite{brandt16_lower_stoc,chang19_exponential_siamjc,%
    ghaffari17_distributed_soda,ghaffari18_derandomizing_focs,balliu23_sinkless_sosa}.
    \item One can construct an (artificial) LCL problem that shows that having
    access to shared randomness helps exponentially in comparison with private
    randomness, and this also gives a separation between e.g.\ \randlcl and
    \rolcl \cite{balliu24_shared_arxiv}.
    \item In 2-dimensional grids, 3-coloring has polynomial locality in e.g.\
    \rslocal and non-signaling models, while it can be solved with logarithmic
    locality in \detolcl \cite{akbari_et_al:LIPIcs.ICALP.2023.10,no_distributed_quantum_advantage}.
    \item In paths, 3-coloring requires $\Theta(\log^*n)$ locality in \randlcl
    \cite{linial92_locality_siamjc,naor91_lower_siamjdm}, while it can be solved
    with $O(1)$ locality in the bounded-dependence model
    \cite{holroyd16_finitely_fmpi}.
\end{itemize}
However, we do not have any such separations in rooted trees outside the $O(\log^* n)$ region. Moreover, in general unrooted trees, all separations are in the sub-logarithmic region. In this work we give a justification for this phenomenon.

\subsection{Contributions}

We study in this work LCL problems in the following three main settings, all
familiar from prior work:
\begin{enumerate}
    \item LCL problems in trees in general,
    \item LCL problems in unrooted regular trees (there are no inputs and we only care about nodes with exactly $d$ neighbors),
    \item LCL problems in rooted regular trees (there are no inputs, edges are oriented, and we only care about nodes with exactly $1$ successor and $d$ predecessors).
\end{enumerate}
In all these settings, it is known that the locality of any LCL problem in
\detlcl falls in one of the following classes: $O(1)$, $\Theta(\log^* n)$,
$\Theta(\log n)$, or $\Theta(n^{1/k})$ for some $k = 1, 2, 3, \dots$
\cite{grunau22_landscape_podc,TimeHierarchyLOCAL,%
chang20_complexity_disc,balliu23_locally_dc,balliu2022efficient}.
Furthermore, in the case of rooted regular trees, we can (relatively
efficiently) also decide which of these classes any given LCL problem belongs to
\cite{balliu2022efficient}.

Our main contribution is showing that, in many of these complexity classes,
\textbf{\boldmath\detlcl and \rolcl are asymptotically equally strong}:
\begin{enumerate}
    \item In general trees, the localities in \rolcl and \detlcl are
    asymptotically equal in the region $\omega(\sqrt{n})$.
    \item In unrooted regular trees, the localities in \rolcl and \detlcl are
    asymptotically equal in the region $\omega(\log n)$.
    \item In rooted regular trees, the localities in \rolcl and \detlcl are
    asymptotically equal in the region $\omega(\log^* n)$.
\end{enumerate}
By prior work, the relation between \detlcl and \rolcl was well-understood in
the $o(\log \log \log n)$ region for rooted (not necessarily regular) trees
\cite{akbari24_online_arxiv}; see \cref{fig:landscape-lcls-before} for the state
of the art before this work.
Putting together prior results and new results, the landscape shown in
\cref{fig:landscape-lcls-after} emerges.

\subsection{Roadmap}

We give an overview of key ideas in \cref{sec:ideas}.
We formally define LCL problems and the models of computing in \cref{sec:background},
and then we are ready to analyze the sub-logarithmic region in rooted regular trees in \cref{sec:rooted-regular-sub-logarithmic}.
Then in \cref{sec:depth} we introduce machinery related to the notion of \emph{depth} from prior work, and equipped with that we can analyze the super-logarithmic region in rooted and unrooted regular trees in \cref{sec:rooted-regular-super-logarithmic,sec:unrooted-regular}.
Finally, we analyze the case of general trees in \cref{sec:general-gap}.

\section{Key Ideas, Technical Overview, and Comparison with Prior Work}\label{sec:ideas}

We keep the discussion at a high level but invite the interested reader to
consult the formal definitions in \cref{sec:background} as needed.

Our results heavily build on prior work that has studied LCL problems in
deterministic and randomized \local models. 
In essence, our goal is to extend its scope across the entire landscape of
models in \cref{fig:models}.

For example, by prior work we know that any LCL problem $\Pi$ in trees has
locality either $O(\sqrt{n})$ or $\Omega(n)$ in the \detlcl model
\cite{almostGlobalLOCAL}. 
Results of this type are known as \emph{gap results}.
For our purposes, it will be helpful to interpret such a gap result as
a \emph{speedup theorem} that allows us to speed up \detlcl algorithms:
\begin{quote}
    If one can solve $\Pi$ with locality $o(n)$ in \textbf{\boldmath\detlcl},
    then one can solve the same problem $\Pi$ with locality $O(\sqrt{n})$ in
    \textbf{\boldmath\detlcl}.
\end{quote}
In essence our objective is to strengthen the statement as follows:
\begin{quote}
    If one can solve $\Pi$ with locality $o(n)$ in \textbf{\boldmath\rolcl},
    then one can solve the same problem $\Pi$ with locality $O(\sqrt{n})$ in
    \textbf{\boldmath\detlcl}.
\end{quote}
The critical implication would be that a faster \rolcl algorithm results not
only in a faster \rolcl algorithm, but also in a faster
\detlcl algorithm---we could not only reduce locality
\enquote{for free} but even switch to a weaker model. 
As a consequence, the complexity class $\Theta(n)$ would contain the same
problems across all models---since $\Omega(n)$ in \rolcl trivially implies
$\Omega(n)$ in \detlcl (which is a weaker model), the above would give us that
$o(n)$ in \rolcl implies $o(n)$ in \detlcl.

If we could prove a similar statement for \emph{every} gap result for LCLs in
trees, we would then have a similar implication for each complexity class: the
complexities across all models in \cref{fig:models} would be the same for every
LCL problem in trees. 
Alas, such a statement cannot be true in full generality (as we discussed above,
there are some known separations between the models), but in this work we take
major steps in many cases where that seems to hold.
To understand how we achieve this, it is useful to make a distinction between
two flavors of prior work: classification- and speedup-style arguments; we
will make use of both techniques.

\subsection{Classification Arguments for Regular Trees}

First, we have prior work that is based on the idea of classification of LCLs, see e.g.\ \cite{balliu2022efficient}. The high-level strategy is as follows:
\begin{enumerate}
    \item Define some property $P$ so that, for any given LCL $\Pi$, we can
    decide whether $\Pi$ has property $P$.
    \item Show that, if $\Pi$ can be solved with locality $o(n)$ in \textbf{\boldmath\detlcl}, then this implies that $\Pi$ must have property $P$.
    \item Show that any problem with property $P$ can be solved with locality $O(\sqrt{n})$ in \textbf{\boldmath\detlcl}.
\end{enumerate}
Such a strategy not only shows that there is a gap in the complexity landscape,
but it also usually gives an efficient strategy for determining what is the
complexity of a given LCL (as we will need to check some set of properties $P_1,
P_2, \dotsc$ and see which of them holds for a given $\Pi$ in order to determine
where we are in the complexity landscape). 
For such results our key idea is to modify the second step as follows:
\begin{enumerate}[1'.]
    \setcounter{enumi}{1}
    \item Show that, if $\Pi$ can be solved with locality $o(n)$ in \textbf{\boldmath\rolcl}, then this implies that $\Pi$ must have property $P$.
\end{enumerate}
Note that we do not need to change the third step; as soon as we establish
property $P$, the pre-existing \detlcl algorithm kicks in.
This is, in essence, what we do in
\cref{sec:rooted-regular-sub-logarithmic,sec:rooted-regular-super-logarithmic,sec:unrooted-regular}:
we take the prior classification from \cite{balliu2022efficient} for rooted and
unrooted regular trees and show that it can be extended to cover the entire
range of models all the way from \detlcl to \rolcl.

\subparagraph*{Log-star certificates.}

One key property that we make use of are \emph{certificates of $O(\log* n)$
solvability} for rooted regular trees, introduced in \cite{balliu23_locally_dc}.
In \cref{sec:rooted-regular-sub-logarithmic}, we show that the existence of an
$o(\log n)$-locality \rolcl algorithm implies that the LCL problem has such a
certificate; in turn, as the name suggests, the certificate of solvability
implies that the problem can be solved with $O(\log^* n)$ locality in both
\congest and \local.
This then extends to an $O(1)$ upper bound in
\slcl~\cite{DBLP:conf/stoc/GhaffariKM17} and in
\detolcl~\cite{akbari_et_al:LIPIcs.ICALP.2023.10}, which then directly implies
the same upper bound also for \rolcl model.
Hence we obtain the following:

\begin{restatable}{theorem}{restateThmRootedSubLogarithmic}
  Let $\Pi$ be an LCL problem on rooted regular trees.
  If $\Pi$ can be solved with $o(\log n)$ locality in \rolcl, then it can be
  solved in \local with $O(\log* n)$ locality.
  Consequently, $\Pi$ can be solved with $O(1)$ locality in \slocal, and thus
  also with the same locality in \onlinelocal, even deterministically.
\end{restatable}

\subparagraph*{Depth.}

Another key property that we make use of is the \emph{depth} of an LCL (see \cref{sec:depth}):
to every LCL problem
$\Pi$ on \emph{regular} trees there is an associated quantity $d_\Pi$ called its
depth that depends only on the description of $\Pi$ and which allows us to
classify its complexity in the \detlcl model \cite{balliu2022efficient}.
We show that depth also captures the complexity in \rolcl.
In \cref{sec:rooted-regular-super-logarithmic}, analyze the case of rooted
regular trees and show the following:

\begin{restatable}{theorem}{restateThmRootedSuperLogarithmic}%
    \label{thm:rooted-full}
    Let $\Pi$ be an LCL problem on rooted regular trees with finite depth $k =
    d_\Pi > 0$.
    Then any algorithm $\mathcal A$ solving $\Pi$ in \rolcl must have locality
    $\Omega(n^{1/k})$.
\end{restatable}

The high-level strategy is as follows:
Assume we have an algorithm with locality~$o(n^{1/k})$.
We construct a large experiment graph and run the algorithm on that.
If the algorithm fails, then it contradicts the assumption that we have such algorithm.
If the algorithm succeeds, then it also produces a certificate showing that~$d_\Pi > k$, which again is a contradiction.
In \cref{sec:unrooted-regular}, we prove an analogous statement for unrooted regular trees:

\begin{restatable}{theorem}{restateThmUnrootedFull}%
  \label{thm:unrooted-full}
  Let $\Pi$ be an LCL problem on unrooted regular trees with finite depth
  $k = d_\Pi > 0$.
  Then any \rolcl algorithm for $\Pi$ has locality $\Omega(n^{1/k})$.
\end{restatable}

\subparagraph*{Putting things together.}

Combining the new results with results from prior work \cite{grunau22_landscape_podc,%
TimeHierarchyLOCAL,chang20_complexity_disc,balliu23_locally_dc,%
balliu2022efficient}, we extend the characterization of LCLs in both unrooted and rooted regular trees all the way up to \rolcl:

\begin{corollary}%
  \label{cor:unrooted_full}
  Let $\Pi$ be an LCL problem on unrooted regular trees.
  Then the following holds:
  \begin{itemize}
    \item If $d_\Pi = 0$, the problem is unsolvable. 
    \item If $0 < d_\Pi < \infty$, then $\Pi$ has locality $\Theta(n^{1/d_\Pi})$
    in both deterministic \local and \rolcl.
    \item If $d_\Pi = \infty$, then $\Pi$ can be solved in \congest -- and hence
    also in \onlinelocal, even deterministically -- with locality $O(\log n)$.
  \end{itemize}
\end{corollary}

\begin{corollary}
  If $\Pi$ is a solvable LCL problem in rooted regular trees, then it belongs to
  one of the following four classes:
  \begin{enumerate}
    \item $O(1)$ in both \detlcl and \rolcl
    \item $\Theta(\log* n)$ in \detlcl and $O(1)$ in \rolcl
    \item $\Theta(\log n)$ in both \detlcl and \rolcl
    \item $\Theta(n^{1/k})$ in both \detlcl and \rolcl where $k = d_\Pi > 0$
  \end{enumerate}
\end{corollary}

\subparagraph*{Comparison with prior work.}

There is prior work \cite{akbari_et_al:LIPIcs.ICALP.2023.10} that took first steps towards classifying LCL problems in rooted regular trees; the main differences are as follows:
\begin{enumerate}
    \item We extend the classification all the way to \rolcl, while
    \cite{akbari_et_al:LIPIcs.ICALP.2023.10} only discusses \detolcl. While this
    may at first seem like a technicality, \cref{fig:models} highlights the
    importance of this distinction: \rolcl captures \emph{all} models in this
    diagram.
    This includes in particular the \qlocal model and the non-signaling model.
    Note that it is currently open if \detolcl captures either of these models.
    \item We build on the classification of \cite{balliu2022efficient}, while \cite{akbari_et_al:LIPIcs.ICALP.2023.10} builds on the (much simpler and weaker) classification of \cite{balliu23_locally_dc}. This has two direct implications:
    \begin{enumerate}
        \item Our work extends to \emph{unrooted trees}; the results in
        \cite{akbari_et_al:LIPIcs.ICALP.2023.10} only apply to rooted trees.
        \item We get an \emph{exact classification} also for complexities
        $\Theta(n^{1/k})$. 
        Meanwhile \cite{akbari_et_al:LIPIcs.ICALP.2023.10}
        essentially only shows that, if the complexity is $\Theta(n^{1/k})$ in
        \detlcl, then it is $\Theta(n^{1/\ell})$ in \onlinelocal for some $\ell
        \le k$; hence it leaves open the possibility that these problems could
        be solved much faster (i.e., $\ell \gg k$) in \onlinelocal.
    \end{enumerate}
\end{enumerate}

\subsection{Speedup Arguments for General Trees}

Second, we have also prior work based on speedup simulation arguments,
see e.g.\ \cite{almostGlobalLOCAL}. 
The high-level strategy there is as follows:
\begin{quote}
    Assume we are given some algorithm $\algo$ that solves $\Pi$ with locality $o(n)$ in \textbf{\boldmath\detlcl}. Then we can construct a new algorithm $\algo'$ that uses $\algo$ as a black box to solve $\Pi$ with locality $O(\sqrt{n})$ in \textbf{\boldmath\detlcl}.
\end{quote}
Unlike classification-style arguments, this does not (directly) yield an
algorithm for determining the complexity of a given LCL problem. In essence,
this can be seen as a black-box simulation of $\algo$. For speedup-style
arguments, the key idea for us to proceed is as follows:
\begin{quote}
    Assume we are given some algorithm $\algo$ that solves $\Pi$ with locality $o(n)$ in \textbf{\boldmath\rolcl}. Then we can construct a new algorithm $\algo'$ that uses $\algo$ as a black box to solve $\Pi$ with locality $O(\sqrt{n})$. Moreover, we can simulate $\algo'$ efficiently in the \textbf{\boldmath\detlcl} model.
\end{quote}
Here a key challenge is that we need to explicitly change models; while \emph{in
general} \detlcl is not strong enough to simulate \rolcl, we show that \emph{in
this case} such a simulation is possible.
This is, in essence, what we do in \cref{sec:general-gap}:
starting with the argument from \cite{almostGlobalLOCAL}, we show that
sublinear-locality \rolcl algorithms not only admit speedup inside \rolcl, but
also a simulation in \detlcl.
Formally, we prove:

\begin{restatable}{theorem}{restateThmUnrootedSqrtn}%
  \label{thm:unrooted-sqrtn}
  Suppose $\Pi$ is an LCL problem that can be solved with $o(n)$ locality in
  \onlinelocal (with or without randomness) in unrooted trees.
  Then $\Pi$ can be solved in unrooted trees in \local with $O(\sqrt{n})$
  locality (with or without randomness, respectively).
\end{restatable}

\section{Preliminaries}
\label{sec:background}

We denote the set of natural numbers, including zero, by $\N$.
For positive integers, excluding zero, we use $\N_+$.
For a set or multiset $X$ and $k \in \N_+$, we write $\multich{X}{k}$ for the
set of multisets over $X$ of cardinality $k$.

All graphs are simple.
Unless stated otherwise, $n$ always denotes the number of nodes in the graph.
An algorithm solving some graph problem is said to succeed with high probability
if, for all graphs except finitely many, it fails with probability at most
$1/n$.

\subsection{Locally Checkable Labeling Problems}

In this section we define locally checkable labeling (LCL) problems.
We first recall the most general definition due to Naor and Stockmeyer~\cite{whatcanbecomputedlocally}:

\begin{definition}[LCL problem in general graphs]%
  \label{def:lcl}
  A \emph{locally checkable labeling} (LCL) problem $\Pi=(\Sigma_\iin,
  \Sigma_\oout, \CC, r)$ is defined as follows:
  \begin{itemize}
    \item $\Sigma_\iin$ and $\Sigma_\oout$ are finite, non-empty sets of input
    and output labels, respectively. 
    \item $r \in \N_+$ is the \emph{checkability radius}.
    \item $\CC$ is the set of \emph{constraints}, namely a finite set of graphs
    where:
    \begin{itemize}
      \item Each graph $H\in\CC$ is centered at some node $v$.
      \item The distance of $v$ from all other nodes in $H$ is at most $r$.
      \item Each node $u\in H$ is labeled with a pair $(i(u), o(u)) \in
      \Sigma_\iin \times \Sigma_\oout$.
    \end{itemize}
  \end{itemize}
  For a graph $G = (V,E)$ whose vertices are labeled according to $\Sigma_\iin$,
  a (node) labeling of a graph $G = (V,E)$ is a \emph{solution} to $\Pi$ if it
  labels every node $v \in V$ with a label from $\Sigma_\oout$ such that the
  $r$-neighborhood of $v$ in $G$ is identical to some graph of $\CC$ (when we
  place $v$ at the center of the respective graph in $\CC$).
  Every node for which this holds is said to be \emph{correctly labeled}.
  If all nodes of $G$ are labeled with a solution, then $G$ itself is
  \emph{correctly labeled}.
\end{definition}

We use this definition to refer to problems in the case of general, that is, not
necessarily regular trees.
For regular trees, we use two other formalisms, one for the case of rooted and
one for that of unrooted trees.
These are simpler to work with and require checking only labels in the immediate
vicinity of a node or half-edge.
We note there is precedent in the literature for taking this approach (see, e.g.,
\cite{balliu2022efficient}).

\begin{definition}[LCL problem on regular rooted trees]%
  \label{def:lcl-rooted}
  An \emph{LCL problem on regular rooted trees} is a tuple $\Pi = (\Delta,
  \Sigma, \mathcal{V})$ where:
  \begin{itemize}
    \item $\Delta \in \N_+$ and $\Sigma$ is a finite, non-empty set.
    \item The (node) \emph{constraints} form a set $\mathcal{V}$ of pairs
    $(\sigma, S)$ where $\sigma \in \Sigma$ and $S \subseteq
    \multich{\Sigma}{\Delta}$.
  \end{itemize}
  For a rooted tree $T = (V,E)$ with nodes having degree in $\{ 1, \Delta \}$, a
  \emph{solution} to $\Pi$ is a labeling $\ell\colon V \to \Sigma$ such that,
  for each node $v$ with $\Delta$ children $u_1,\dots,u_\Delta$, we have
  $(\ell(v), \{\ell(u_1),\dots,\ell(u_\Delta)\}) \in \mathcal{V}$.
  Again, this requirement allows us to speak of \emph{correctly labeled} nodes
  and thus also \emph{correctly labeled} trees.
\end{definition}

Under this formalism, leaves may be labeled arbitrarily.
When $\Pi$ is clear from the context, we refer to the elements of
$\multich{\Sigma}{\Delta}$ as node configurations.

Meanwhile, in the case of unrooted trees, we work with LCL problems where the
labels are placed on half-edges (instead of nodes) and are subject to node and edge
constraints.

\begin{definition}[LCL problem on regular unrooted trees]%
  \label{def:lcl-unrooted}
  An \emph{LCL problem on unrooted trees} is a tuple $\Pi = (\Delta, \Sigma,
  \mathcal{V}, \mathcal{E})$ where:
  \begin{itemize}
    \item $\Delta \in \N_+$ and $\Sigma$ is a finite, non-empty set.
    \item The \emph{node constraints} form a set $\mathcal{V} \subseteq
    \multich{\Sigma}{\Delta}$.
    \item The \emph{edge constraints} form a set $\mathcal{E} \subseteq
    \multich{\Sigma}{2}$.
  \end{itemize}
  For an unrooted tree $T = (V,E)$ with nodes having degree in $\{ 1, \Delta
  \}$, a \emph{solution} to $\Pi$ is a labeling of half-edges, that is, a map
  $\ell\colon V \times E \to \Sigma$ such that the following holds:
  \begin{itemize}
    \item For every node $v \in V$ with degree $\Delta$ that is incident to the
    edges $e_1,\dots,e_\Delta$, we have $\{ \ell(v,e_1),\dots,\ell(v,e_\Delta)
    \} \in \mathcal{V}$.
    \item For every edge $e = \{ u,v \} \in E$, $\{ \ell(u,e),\ell(v,e) \} \in
    \mathcal{E}$.
  \end{itemize}
  As before, if $T$ is labeled with a solution, then it is \emph{correctly
  labeled}.
  In light of these two types of requirements, it is also natural to speak of
  \emph{correctly labeled} nodes and edges.
\end{definition}

As in the preceding definition, leaves are unconstrained and may
be labeled arbitrarily.
As before, if such an LCL problem $\Pi$ on regular unrooted trees is clear from
the context, we refer to the elements of $\multich{\Sigma}{\Delta}$ as node
configurations and to those of $\multich{\Sigma}{2}$ as edge configurations.

\subsection{Models of Distributed Computing}
Next we formally define the \local model and its extensions.

\begin{definition}[Deterministic \local model]
    The deterministic \local model of distributed computing runs on a graph $G =
    (V,E)$ where each node $v\in V$ represents a computer and each edge a
    connection channel.
    Each node is labeled with a unique identifier.
    All nodes run the same distributed algorithm in parallel.
    Initially, a node is only aware of its own identifier and degree. 
    Computation proceeds in synchronous rounds, and in each round a node can send and receive a message to and from each
    neighbor and update its state. Eventually each node must stop and announce its local output (its part of
    the solution, e.g. in graph coloring its own color). The \emph{running time}, \emph{round complexity}, or \emph{locality} of the algorithm is the (worst-case) number of rounds $T(n)$ until the algorithm stops in any $n$-node graph.
\end{definition}

\begin{definition}[\Randlcl model]
    The \randlcl model is defined identically to the (deterministic) \local
    model, with the addition that each node has access to a private, infinite
    stream of random bits.
    Additionally, a \randlcl algorithm is required to succeed with high
    probability.
\end{definition}

\begin{definition}[\Detolcl model \cite{akbari_et_al:LIPIcs.ICALP.2023.10}]
    In the (deterministic) \onlinelocal model, an adversary chooses a sequence of nodes, $\sigma = (v_1, v_2, \dots, v_n)$, and reveals the nodes to the algorithm one at a time.
    The algorithm processes each node sequentially. Given an \onlinelocal algorithm with locality $T(n)$, when a node $v_i$ is revealed, the algorithm must choose an output for $v_i$
    based on the subgraph induced by the radius-$T(n)$ neighborhoods of $v_1, v_2, \dots, v_i$. In other words, the algorithm retains global memory.
\end{definition}

\begin{definition}[\Rolcl model \cite{akbari24_online_arxiv}]
    In the \rolcl model, an adversary first commits to a sequence of nodes,
    $\sigma = (v_1, v_2, \dots, v_n)$, and reveals the nodes to the algorithm
    one at a time, as in the \onlinelocal model.
    The algorithm runs on each $v_i$ and retains global memory, as in the
    \onlinelocal model. 
    Additionally, the algorithm has access to an infinite stream of random bits
    unbeknownst to the adversary; that is, the adversary \emph{cannot} change
    the order they present the remaining nodes based on the intermediate output
    of the algorithm.
    As in the \randlcl, the algorithm is required to succeed with high
    probability.
\end{definition}

\section{Sub-logarithmic Gap for LCLs in Rooted Regular Trees}
\label{sec:rooted-regular-sub-logarithmic}

In this section, we show that there are no LCL problems on rooted trees with
locality in the range $\omega(1) \text{---} o(\log n)$ in the \rolcl model.
Moreover, the problems that are solvable with locality $O(1)$ in \rolcl are exactly the same problems that are solvable with locality $O(\log* n)$ in the \local model.

\restateThmRootedSubLogarithmic*

We show this by showing that a locality-$o(\log n)$ \rolcl algorithm solving LCL problem~$\Pi$ implies the existence of a \emph{coprime certificate for $O(\log* n)$ solvability} (see \cref{def:rooted-tree-logstar-certificate}) that then implies that $\Pi$ is solvable with locality $O(\log* n)$ in the \local model~\cite{balliu23_locally_dc}, and hence with locality $O(1)$ in the \rolcl model.

Throughout this section, we consider an LCL problem~$\Pi = (\Delta, \Sigma, \VV)$ and a \rolcl algorithm~$\algo$ that solves~$\Pi$ with locality~$T(n) = o(\log n)$ with high probability.
We start by constructing a family of input instances and then argue that $\algo$
solving these instances produces a canonical labeling regardless of the
randomness.
Finally, we show how this canonical labeling yields the coprime certificate for $O(\log* n)$ solvability.

\subsection{Construction of Input Instances}

Let $d$ be a depth parameter that we fix later, and let $\Delta$ be the number of children of internal nodes.
We now construct the input instance as follows:
\begin{definition}[Family of input instances]
    \label{def:input-instance}
    To construct the family of input instances:
    \begin{enumerate}
        \item Construct $(\abs{\Sigma} + 1)$ chunks of trees, each containing $\Delta^{d+1}$ complete rooted trees of height~$2d$.
        Give the trees in each chunk an ordering.
        Let $M$ be the set of nodes in these trees at distance~$d$ from the root; we call these nodes \emph{middle nodes}.
        \item Choose a bit~$b$, and choose a node~$u$ which is at depth $d$ in any of the trees created in the previous phase.
        Choose a chunk~$C$ such that node~$u$ is not contained in chunk~$C$.
        \item The subtree rooted at $u$ has $\Delta^d$ leaf descendants.
        \begin{description}
            \item[If $b=0$:] Identify the roots of the first $\Delta^d$ trees of chunk~$C$ with the leaf descendants of $u$ in a consistent order (see bottom-left of \cref{fig:0uC1uC}).
            \item[If $b=1$:] Make the roots of the first $\Delta^{d+1}$ trees of chunk~$C$ the children of the leaf descendants of $u$ in a consistent order (see bottom-right \cref{fig:0uC1uC}).
        \end{description}
    \end{enumerate}
\end{definition}

Note that trees constructed in this way have $n = (\abs{\Sigma} + 1) \cdot \Delta^{d+1} \cdot \frac{\Delta^{2d + 1} - 1}{\Delta - 1}$ nodes.
Moreover, choosing $(b, u, C)$ uniquely fixes the construction.
Let $\PP$ be the set of choices for $(b, u, C)$.

\begin{observation}
    \label{obs:instance-card-lt-n}
    The number of instances is less than the number of nodes in each instance.
    In particular, $\abs{\PP} = 2 \abs{\Sigma} (\abs{\Sigma} + 1) \Delta^{2d+1} < n$.
\end{observation}

We can now fix depth parameter~$d$ such that $d > \log_\Delta (2 \abs{\Sigma})$ and $d > T(n)$.
Recall that $T(n) = o(\log n)$, and hence such $d$ exists.

\begin{figure}[t!]
    \centering
    \includegraphics[width=\textwidth]{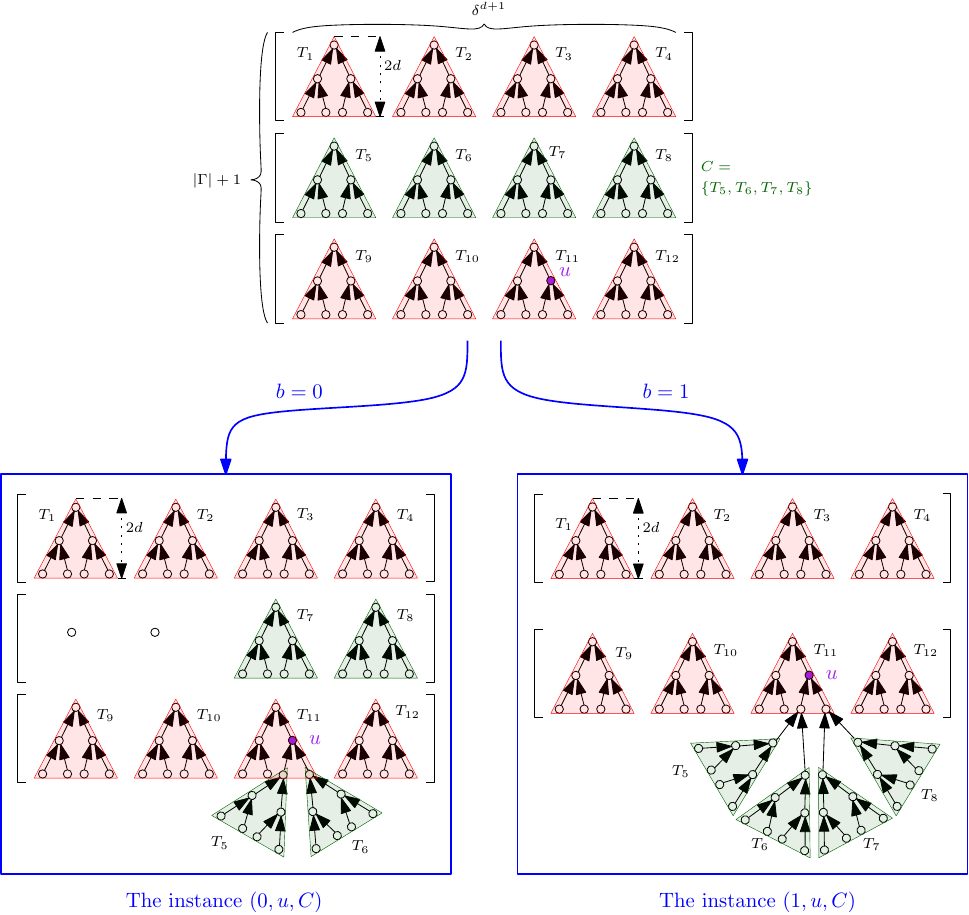}
    \caption{Visualization of two trees in the family of input instances. At the top we have $\abs{\Sigma}+1$ chunks of complete trees of depth~$2d$, each containing $\Delta^{d+1}$ trees. The green trees in the middle row represent the chosen chunk~$C$. On the bottom-left, we have $b=0$, in which case we identify the roots of the first $\Delta^d$ trees of~$C$ with the leaf descendants of node~$u$. On the bottom-right, we have $b=1$, in which case we make all trees in chunk~$C$ the children of the leaf descendants of node~$u$. See \cref{def:input-instance} for full construction.}
    \label{fig:0uC1uC}
\end{figure}

\subsection{Randomness of the Algorithm}

Throughout this discussion, we let $\Pr_{R}(\cdot)$ denote the probability of an event when the randomness is over some source $R$ for which the probability is being defined.

For each instance $(b, u, C)$, we fix the order middle nodes~$M$ such that it is consistent across all instances.
We then reveal these middle nodes to~$\algo$ in this order.
Let $S \in \Sigmabar = \Sigma^{\Delta^{2d+1}(\abs{\Sigma}+1)}$ be the random sequence of labels generated by $\algo$ for the middle nodes~$M$.
Since the locality of~$\algo$ is $T(n) < d$, the algorithm does not get to see the roots or the leaves of these height-$2d$ trees.
In particular, the algorithm does not get to know which instance $(b, u, C)$ we have chosen.
Hence $S$ must be independent of the choice of $(b, u, C)$.

\begin{observation}
    \label{obs:S-indp}
    For every choice $\sigmabar \in \Sigmabar$ of labeling middle nodes~$M$, and every instance $(b, u, C) \in \PP$, we must have
    \begin{equation*}
        \Pr_S(S = \sigmabar) = \Pr_S(S = \sigmabar \mid (b, u, C))
    \end{equation*}
    where $\Pr_S(S = \sigmabar)$ denotes the probability that $\algo$ produces output $\sigmabar$ for nodes~$M$, and $\Pr_S(S = \sigmabar \mid (b, u, C))$ denotes the same given that the input instance was~$(b, u, C)$.
\end{observation}

Note that while the output~$S$ for nodes~$M$ is independent of $(b, u, C)$, the output that $\algo$ produces for the rest of the nodes of the input may be dependent of $(b, u, C)$.
We denote this random variable by $\Abuc$.

We now analyze the failure probability of~$\algo$ to show that there exists a fixed labeling~$\sigmabaropt$ for middle nodes~$M$ such that the labeling is still completable for each input instance $(b, u, C)$.
Let $\Afails$ denote the event that algorithm~$\algo$ fails to solve the problem.
By assumption, $\algo$ succeeds with high probability; hence $\Pr_{S, \Abuc}(\Afails) \le \frac{1}{n}$ for all choices of $(b, u, C) \in \PP$.
We can now pick a labeling of the middle nodes that minimizes the average failure probability:

\begin{definition}
    Let $\sigmabaropt$ be defined as follows:
    \begin{equation*}
        \sigmabaropt \in
        \argmin_{\sigmabar \in \Sigmabar, \Pr_S(S = \sigmabar) > 0} 
        \sum_{(b, u, C) \in \PP} \Pr_{\Abuc}(\Afails \mid (b, u, C), S = \sigmabar)
    \end{equation*}
    where ties are broken deterministically.
\end{definition}

This definition ensures the following:
Give that the algorithm labels nodes~$M$ with $\sigmabaropt$, the total probability of failure over all instances $(b, u, C)$ is minimized.
In some sense, $\sigmabaropt$ is the best labeling for~$M$ when the algorithm know nothing about the instance.
Note that~$\sigmabaropt$ is a concrete element of $\Sigmabar$ and hence does not depend on the randomness of the algorithm.
We can formalize this idea:

\begin{lemma}
    \label{lem:sol-always-exists}
    Regardless of the instance $(b, u, C) \in \PP$, if the labeling~$S$ of nodes~$M$ is $\sigmabaropt$, there exists a valid way to label the remaining nodes of the instance.
\end{lemma}

\begin{proof}
    We start by noting that
    \begin{equation*}
        \frac{1}{n}
        \ge
        \Pr_{S, \Abuc}(\Afails)
    \end{equation*}
    since $\algo$ works correctly with high probability.
    We can now expand the probability to be conditional over the initial labeling~$S$:
    \begin{equation*}
        \Pr_{S, \Abuc}(\Afails)
        =
        \sum_{\sigmabar \in \Sigmabar} \Pr_{\Abuc} (\Afails \mid (b, u, C), S = \sigmabar) \Pr_S(S = \sigmabar \mid (b, u, C)) .
    \end{equation*}
    Next, we apply \cref{obs:S-indp} to get
    \begin{equation*}
        \frac{1}{n}
        \ge
        \sum_{\sigmabar \in \Sigmabar} \Pr_{\Abuc} (\Afails \mid (b, u, C), S = \sigmabar) \Pr_S(S = \sigmabar) ,
    \end{equation*}
    and then we sum over all choices of $(b, u, C)$ on both sides to get
    \begin{equation*}
        \frac{\abs{\PP}}{n} = \sum_{(b, u, C) \in \PP} \frac{1}{n}
        \ge
        \sum_{(b, u, C) \in \PP} \sum_{\sigmabar \in \Sigmabar} \Pr_{\Abuc} (\Afails \mid (b, u, C), S = \sigmabar) \Pr_S(S = \sigmabar)
    \end{equation*}
    Exchanging the order of summation and noting that $\sigmabaropt$ minimizes \[\sum_{(b, u, C) \in \PP} \Pr_{\Abuc} (\Afails \mid (b, u, C), S = \sigmabar),\] we get
    \begin{equation*}
        \frac{\abs{\PP}}{n}
        \ge
        \sum_{(b, u, C) \in \PP} \Pr_{\Abuc} (\Afails \mid (b, u, C), S = \sigmabaropt).
    \end{equation*}
    Recalling that $\frac{\abs{\PP}}{n} < 1$ by \cref{obs:instance-card-lt-n} and that all probabilities are non-negative numbers, we complete our calculation:
    \begin{equation*}
        1 > \Pr_{\Abuc} (\Afails \mid (b, u, C), S = \sigmabaropt) \quad \text{for each } (b, u, C) \in \PP.
    \end{equation*}
    Hence for each instance $(b, u, C)$, the algorithm fails with probability
    strictly less than 1.
    Thus it succeeds with non-zero probability, and hence a correct labeling must exist.
\end{proof}

\subsection{Getting a Coprime Certificate as a Subset of All Such Instances}

We are now ready to extract from algorithm~$\algo$ a coprime certificate for $O(\log* n)$ solvability.
The idea is to pick a subset of input instances~$\PP$ that---along with the labeling~$\sigmabaropt$---form the pairs of sequences of trees of the certificate.
We defer the proof of why this implies the existence of a locality-$O(\log* n)$ \local algorithm to previous work~\cite{balliu23_locally_dc}.

\begin{definition}[Coprime certificate for $O(\log* n)$ solvability \cite{balliu23_locally_dc}]
    \label{def:rooted-tree-logstar-certificate}
    Let $\Pi = (\Delta, \Sigma, C)$ be an LCL problem. A certificate for $O(\log* n)$ solvability of $\Pi$ consists of labels $\Sigma_{\mathcal{T}} = \{\sigma_1, \ldots, \sigma_t\} \subseteq \Sigma$, a depth pair $(d_1, d_2)$ and a pair of sequences $\mathcal{T}^1$ and $\mathcal{T}^2$ of $t$~labeled trees such that
    \begin{enumerate}
        \item The depths $d_1$ and $d_2$ are coprime.
        
        \item Each tree of $\mathcal{T}^1$ (resp. $\mathcal{T}^2$) is a complete $\Delta$-ary tree of depth $d_1 \ge 1$ (resp. $d_2 \ge 1$).
        
        \item Each tree is labeled by labels from $\Sigma$ correctly according to problem $\Pi$.
        
        \item Let $\bar{\mathcal{T}}^1_i$ (resp. $\bar{\mathcal{T}}^2_i$) be the tree obtained by starting from $\mathcal{T}^1_i$ (resp. $\mathcal{T}^2_i$) and removing the labels of all non-leaf nodes. It must hold that all trees $\bar{\mathcal{T}}^1_i$ (resp. $\bar{\mathcal{T}}^2_i$) are isomorphic, preserving the labeling. All the labels of the leaves of $\bar{\mathcal{T}}^1_i$ (resp. $\bar{\mathcal{T}}^2_i$) must be from set $\Sigma_{\mathcal{T}}$.

        \item The root of tree $\mathcal{T}^1_i$ (resp. $\mathcal{T}^2_i$) is labeled with label $\sigma_i$.
    \end{enumerate}
\end{definition}

Let $\Sigmacert = \{\sigma_1, \sigma_2, \ldots, \sigma_t\} \subseteq \Sigma$ be the set of labels appearing in $\sigmabaropt$, and let $(u_1, u_2, \ldots, u_t)$ be some middle nodes in the input instances such that node~$u_i$ has label~$\sigma_i$ according to $\sigmabaropt$.
Note that there are $\abs{\Sigma} + 1$ chunks whereas $t = \abs{\Sigmacert} < \abs{\Sigma} + 1$. 
Therefore, we get the following result by the pigeonhole principle:
\begin{observation}\label{obs:C_0}
    There is a chunk $C_0$ which does not contain any node in $\{u_1, \ldots, u_t\}$.
\end{observation}

With these definitions of labels~$\Sigmacert = (\sigma_1, \ldots, \sigma_t)$, chunk~$C_0$, and nodes~$(u_1, \ldots, u_t)$, we are ready to prove our main result:

\begin{theorem}\label{thm:cert-exists}
    For an LCL problem $\Pi$ without inputs on rooted regular trees, if there is a \rolcl algorithm $\mathcal A$ with locality $T(n) = o(\log n)$, then there exists a coprime certificate of $O(\log* n)$ solvability for $\Pi$, with $\Sigmacert$ as the subset of labels. 
\end{theorem}
\begin{proof}
    For each $i \in \{1, \ldots, t\}$, consider the instance $(0, u_i, C_0)$ with $S=\sigmabaropt$ and some valid way to label the remaining nodes; such a labeling exists due to~\cref{lem:sol-always-exists}.
    The subtree rooted at $u_i$ contains the first $\Delta^{d}$ trees of the chunk $C_0$.
    Let the depth-$2d$ subgraph rooted at $u_i$ be tree~$\mathcal T^1_i$.
    See \cref{fig:cert2d} for a visualization.

    Note that for each $i \in \{1, \ldots, t\}$, tree~$\mathcal T^1_i$ has a root labeled with $\sigma_i$, and the leaves are labeled with labels from set~$\Sigmacert$ in an identical way.
    Hence $\mathcal T^1$ form the first sequence of the certificate.
    
    Similarly, we consider $(1, u_i, C_0)$ for all $i \in \{1, \ldots, t\}$.
    This time we let the depth-$(2d+1)$ subtree rooted at $u_i$ be tree~$\mathcal T^2_i$.
    See \cref{fig:cert2d+1} for a visualization.
    The sequence $\mathcal T^2$ forms the second sequence of the certificate, again by similar arguments.

    It is easy to check that $\Sigmacert$, $2d$ and $2d+1$ and sequences $\mathcal T^1$ and $\mathcal T^2$ indeed form a coprime certificate for $O(\log* n)$ solvability for problem~$\Pi$. \qedhere

    \begin{figure}[!t]  
        \centering    
        \begin{subfigure}[t]{.39\textwidth}
            \centering
            \includegraphics[scale=.8]{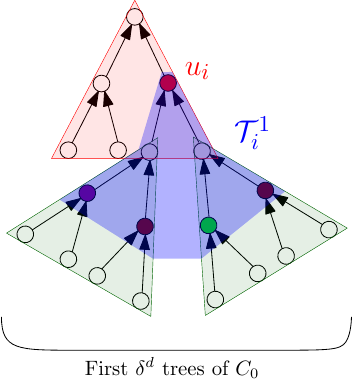}
            \caption{Trees $\mathcal T_i^1$ of depth $2d$ for certificate from instance
            $(0, u_i, C_0)$.}
            \label{fig:cert2d}
        \end{subfigure}
        \hfill
        \begin{subfigure}[t]{.59\textwidth}
            \centering
            \includegraphics[scale=.8]{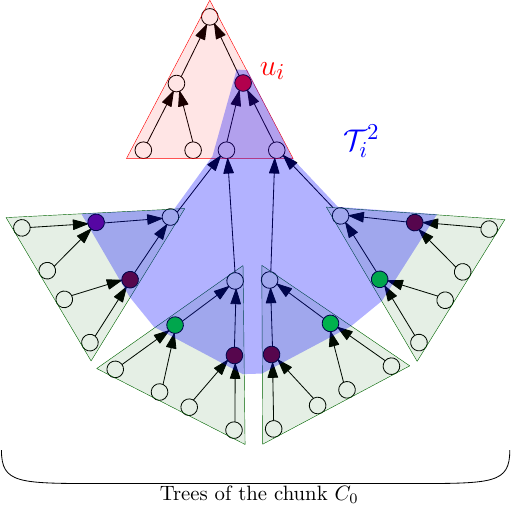}
            \caption{Trees $\mathcal T_i^1$ of depth $2d+1$ for certificate from
            instance $(1, u_i, C_0)$.}
            \label{fig:cert2d+1}
        \end{subfigure}
        \caption{Construction of trees of coprime certificate for $O(\log* n)$ solvability. The blue shaded region represents the tree of the sequence and extends from the labeled middle node~$u_i$ to the middle nodes of the trees hanging from its leaf descendants.}
    \end{figure}
\end{proof}

We can now proof the main result of this section:
\restateThmRootedSubLogarithmic*

\begin{proof}
    A locality-$o(\log n)$ \rolcl algorithm for problem~$\Pi$ implies the existence of certificate of $O(\log* n)$ solvability of~$\Pi$ by \cref{thm:cert-exists}.
    This further implies existence of \local algorithm solving $\Pi$ with locality $O(\log* n)$~\cite{balliu23_locally_dc}, which in turn implies the existence of \slocal algorithm solving $\Pi$ with locality $O(1)$~\cite{DBLP:conf/stoc/GhaffariKM17}.
    Since the \rolcl is stronger than the \slocal model \cite{akbari24_online_arxiv}, this implies the existence of a locality-$O(1)$ \rolcl algorithm that solves $\Pi$.
    Therefore, there can be no LCL problem on rooted regular trees where the optimal algorithm has a locality of $T(n)$ which is both $o(\log n)$ and $\omega(1)$.
\end{proof}

%\bibliography{paper}

%\newpage
%\appendix

\section{Definition of Depth}\label{sec:depth}

Before we continue, we recall notions pertaining to the analysis of
LCL problems in the regular trees case (for both rooted and unrooted trees).
We stress the tools that we present here concern only the \emph{description} of
the problems and thus may be applied no matter what is our computational model
of interest.
Most definitions are morally the same for both rooted and unrooted trees (and,
indeed, are even named the same), though the latter turns out to be more
convoluted.
Hence we will always introduce the relevant concept first for rooted trees.
In fact, to obtain a clearer understanding of the concepts involved, an
unfamiliar reader may prefer to postpone the parts concerning unrooted trees
to a reread of the text altogether.

\subparagraph*{Automata.}

The definitions of LCL problems we introduced above are useful in that they
allow us to associate an LCL problem $\Pi$ with an \emph{automaton} that encodes
correct solutions to $\Pi$.
Conceptually speaking, the automaton describes the labels that we observe as we
are traversing a path in a correctly labeled tree.
For rooted trees this is relatively straightforward to specify; in the unrooted
case we do not have a sense of direction and hence we must encode that in the
nodes by using two components, one specifying which label we see when entering
and the other which label we see when exiting the node.
Alternatively, it may also be instructive to imagine the automaton as describing
paths on the \emph{line graph} of the original tree (i.e., we are traversing the
tree \emph{edge by edge} instead of node by node).

\begin{definition}[Automaton for an LCL on regular trees]%
  \label{def:lcl-automaton}
  Let $\Pi$ be an LCL problem on regular trees.
  \begin{itemize}
    \item If $\Pi = (\Delta, \Sigma, \mathcal{V})$ is a problem on rooted trees,
    then the \emph{automaton} $\mathcal{M}_\Pi$ associated with $\Pi$ is the
    digraph with $\Sigma$ as its set of nodes and where we have an edge
    $(\sigma,\sigma')$ if and only if there is $(\sigma,S) \in \mathcal{V}$ such
    that $\sigma' \in S$.
    \item If $\Pi = (\Delta, \Sigma, \mathcal{V}, \mathcal{E})$ is a problem on
    unrooted trees, then the \emph{automaton} $\mathcal{M}_\Pi$ associated with
    $\Pi$ is the digraph defined as follows:
    \begin{itemize}
      \item The nodes of $\mathcal{M}_\Pi$ are the elements $(x,y) \in \Sigma
      \times \Sigma$ for which $\{ x,y \} \in \multich{\Sigma}{2}$ is a
      sub-multiset of some node configuration in $\mathcal{V}$.
      \item There is an edge between $(x_1,x_2)$ and $(y_1,y_2)$ in
      $\mathcal{M}_\Pi$ if and only if $\{ x_2,y_1 \} \in \mathcal{E}$.
    \end{itemize}
  \end{itemize}
\end{definition}

\subparagraph*{Preliminaries.}

For every LCL problem $\Pi$ on regular trees, there is a fixed quantity $d_\Pi$
that depends only on the problem description and which essentially captures its
complexity.
This quantity is called the \emph{depth} of $\Pi$.
The precise definition is rather involved, so here we attempt to keep it brief
and self-contained; for a more extensive treatment, see for instance
\cite{balliu2022efficient}.
As mentioned before, the definitions for rooted and unrooted trees differ in the
details but are otherwise conceptually very similar to one other.

Let us now fix an LCL problem $\Pi$ on regular trees.
The following is our roadmap in order to define the depth $d_\Pi$:
\begin{enumerate}
  \item First we define what is meant by taking the restriction of $\Pi$ to a
  set of labels (in the rooted case) or set of node configurations (in the
  unrooted case).
  \item Next we define two operations $\trim$ and $\flexSCC$ that induce two
  different kinds of restrictions that we may take.
  \item By alternating between the two, we obtain so-called good sequences of
  problem restrictions.
  \item Finally, $d_\Pi$ is defined to be exactly the maximum length of such a
  good sequence.
\end{enumerate}

\subparagraph*{Restrictions.}
In the case of rooted trees, a restriction is defined in terms of a
\emph{subset} of labels; that is, we are simply considering the same problem but
allowing only a subset of labels to be used in the solution.
In contrast, in unrooted trees, restrictions are defined by a \emph{set of
permissible pairs} of labels and we allow only node configurations to be used
where \emph{every} pair of labels present in the configuration is permissible.

\begin{definition}[Restriction of an LCL on regular trees]%
  \label{def:lcl-restriction}
  Let $\Pi$ be an LCL on regular trees. 
  \begin{itemize}
    \item When $\Pi = (\Delta, \Sigma, \mathcal{V})$ is a problem on rooted
    trees, the \emph{restriction} of $\Pi$ to a subset $\Sigma' \subseteq
    \Sigma$ of labels is the problem $\rest{\Pi}{\Sigma'} = (\Delta, \Sigma',
    \mathcal{V}')$ where $\mathcal{V}'$ consists of all pairs $(\sigma, S) \in
    \mathcal{V}$ for which $\sigma \in \Sigma'$ and $S \subseteq \Sigma'$.
    \item If $\Pi = (\Delta, \Sigma, \mathcal{V}, \mathcal{E})$ is a problem on
    unrooted trees, then the \emph{restriction} of $\Pi$ to $\mathcal{D}
    \subseteq \multich{\Sigma}{2}$ is the problem $\rest{\Pi}{\mathcal{D}} =
    (\Delta, \Sigma, \mathcal{V}', \mathcal{E})$ where $\mathcal{V}' \subseteq
    \mathcal{V}$ is maximal such that $\multich{C}{2} \subseteq \mathcal{D}$
    holds for every $C \in \mathcal{V}'$.
  \end{itemize}
\end{definition}

\subparagraph*{\boldmath The $\trim$ operation.}
Next we define the $\trim$ operation, which restricts the set of labels (or, in
the case of unrooted trees, node configurations) to those that can be used
at the root of a complete tree of arbitrary depth.
Although this sounds like a precise definition, we stress that \enquote{complete
tree} is in fact an ambiguous term in this context since it has different
meanings depending on which LCL formalism we are (i.e., rooted or unrooted
trees).
We settle this matter upfront in a separate definition before turning to that of
$\trim$ proper.

\begin{definition}[Complete regular tree]%
  \label{def:complete-tree}
  Let $i \in \N_+$.
  With $T_i$ we denote the tree of depth $i$ where the root has degree $\Delta -
  1$ and every other inner vertex has degree $\Delta$.
  Meanwhile, $T_i^\ast$ denotes the tree of depth $i$ where every inner vertex
  (including the root) has degree $\Delta$.
\end{definition}

Note that, technically, $T_i$ is not regular.
Nevertheless we can add a single parent node of degree $1$ to the root in order
to make it so.
(Since it has degree $1$, this new node is unconstrained.)
We avoid doing so in order to be able to refer to the root of $T_i$ more easily.

\begin{definition}[$\trim$]%
  \label{def:trim}
  Let $\Pi$ be an LCL on regular trees of degree $\Delta$. 
  \begin{itemize}
    \item If $\Pi = (\Delta, \Sigma, \mathcal{V})$ is a problem on rooted trees,
    then $\trim(\cdot)$ maps subsets of $\Sigma$ again to subsets of $\Sigma$.
    For $\Sigma' \subseteq \Sigma$, $\trim(\Sigma')$ is the set of all $\sigma
    \in \Sigma'$ for which, for every $i \in \N_+$, there is a solution of $\Pi$
    to $T_i$ such that:
    \begin{itemize}
      \item The root is labeled with $\sigma$.
      \item Every other vertex (including the leaves) is labeled with a label
      from $\Sigma'$.
    \end{itemize}
    \item When $\Pi = (\Delta, \Sigma, \mathcal{V}, \mathcal{E})$ is a problem
    on unrooted trees, $\trim(\cdot)$ maps subsets of $\mathcal{V}$ again to
    subsets of $\mathcal{V}$.
    Namely, for $\mathcal{D} \subseteq \mathcal{V}$, $\trim(\mathcal{D})$ is the
    set of node configurations $C \in \mathcal{D}$ for which, for every $i \in
    \N_+$, there is a solution of $\Pi$ to $T_i^\ast$ such that:
    \begin{itemize}
      \item The root has $C$ as its node configuration.
      \item Every other vertex (except for the leaves, which are always
      unrestricted) has an element of $\mathcal{D}$ as its node configuration.
    \end{itemize}
  \end{itemize}
\end{definition}

Hence $\trim$ restricts the sets of labels (or node configurations, in the case
of unrooted trees) that can be placed at the root of the tree.
Note that the rest of the tree does not need to be labeled according to what
ends up being placed in the $\trim$ set but only to the set $\Sigma'$ of labels
(or set $\mathcal{D}$ of node configurations) that we started with.

\subparagraph*{\boldmath The $\flexSCC$ operation.}

In contrast to the $\trim$ operation, which is defined on the basis of labelings
of the complete tree, $\flexSCC$ is computed solely by analyzing (the graph of)
the automaton $\mathcal{M}_\Pi$ associated with $\Pi$.

For concreteness, let us focus on the case of rooted trees.
Given some subset $\Sigma' \subseteq \Sigma$ of labels, we start by considering
the restriction of $\Pi$ to $\Sigma'$, in which case we obtain a restricted LCL
problem $\Pi' = \rest{\Pi}{\Sigma'}$.
Notice the automaton $\mathcal{M}_{\Pi'}$ associated with $\Pi$ is contained in
$\mathcal{M}_\Pi$.
We are interested in the strongly connected components of $\mathcal{M}_{\Pi'}$
and in particular those whose nodes can be reached from one another in a
\emph{flexible} way.
Flexibility here is meant in the sense of the walk's length, namely that one may
reach any vertex from any other one in any desired (large enough) number of
steps.
Conceptually, each such component of $\mathcal{M}_{\Pi'}$ gives us a strategy
for locally filling in labels for arbitrarily long paths using just the set
$\Lambda \subseteq \Sigma'$ of labels in the component.

This notion of flexibility is specified in the next definition.

\begin{definition}[Path-flexibility]%
  \label{def:path-flexible}
  Let $\Pi$ be an LCL problem on regular trees (of either kind), and let
  $\mathcal{M}_\Pi$ be the automaton associated with it.
  A subset $U$ of vertices of $\mathcal{M}_\Pi$ is said to be
  \emph{path-flexible} if there is a constant $K \in \N_+$ such that, for any
  choice of vertices $u,v \in U$, there is a walk of length $k \ge K$ from $u$
  to $v$ that only visits vertices in $U$.
\end{definition}

With this notion we now define $\flexSCC$.
As already stressed several times, the definition for unrooted trees is more
involved since we have to manipulate multisets, but morally it yields the same
operation as in the rooted case.

\begin{definition}[$\flexSCC$]%
  \label{def:flexscc}
  Let $\Pi$ be an LCL on regular trees. 
  \begin{itemize}
    \item If $\Pi = (\Delta, \Sigma, \mathcal{V})$ is a problem on rooted trees,
    then $\flexSCC(\cdot)$ maps subsets of $\Sigma$ to sets of subsets of
    $\Sigma$.
    Namely, for $\Sigma' \subseteq \Sigma$, $\flexSCC(\Sigma')$ is obtained as
    follows: 
    First take the restriction $\Pi' = \rest{\Pi}{\Sigma'}$ of $\Pi$ to
    $\Sigma'$ and construct the automaton $\mathcal{M}_{\Pi'}$ associated with
    it.
    Then $\flexSCC(\Sigma')$ is the set of strongly connected components of
    $\mathcal{M}_{\Pi'}$ that are path-flexible.
    \item When $\Pi = (\Delta, \Sigma, \mathcal{V}, \mathcal{E})$ is a problem
    on unrooted trees, $\flexSCC(\cdot)$ maps subsets of $\mathcal{V}$ to sets
    of elements of $\multich{\Sigma}{2}$.
    For $\mathcal{V}' \subseteq \mathcal{V}$, $\flexSCC(\mathcal{V}')$ is
    constructed as follows:
    First determine the restriction $\Pi' = \rest{\Pi}{\mathcal{V}'}$ of $\Pi$
    to $\mathcal{V}'$ as well as the automaton $\mathcal{M}_{\Pi'}$ associated
    with it.
    Then $\flexSCC(\mathcal{V}')$ is \enquote{morally} the set of strongly
    connected components of $\mathcal{M}_{\Pi'}$ that are path-flexible, but
    lifted back to a set over elements of $\multich{\Sigma}{2}$ (instead of
    elements of $\Sigma \times \Sigma$).
    Concretely we have that $\flexSCC(\mathcal{V}')$ is the set that contains
    every $\mathcal{D} \in \multich{\Sigma}{2}$ for which $U_{\mathcal{D}} = \{
    (x,y) \in V(\mathcal{M}_\Pi) \mid \{ x,y \} \in \mathcal{D} \}$ is a
    path-flexible component in $\mathcal{M}_{\Pi'}$.
  \end{itemize}
\end{definition}

\subparagraph*{Good sequences and depth.}
With the above we may define so-called \emph{good sequences} that are obtained
by applying the two operations $\trim$ and $\flexSCC$ alternatingly.

\begin{definition}[Good sequences]%
  \label{def:good-sequence}
  Let $\Pi$ be an LCL problem on regular trees.
  \begin{itemize}
    \item When $\Pi = (\Delta, \Sigma, \mathcal{V})$ is a problem on rooted
    trees, a \emph{good sequence} is a tuple $s = (\Sigma_1^{\sR},
    \Sigma_1^{\sC}, \dots, \Sigma_k^{\sR})$ of subsets of $\Sigma$ where:
    \begin{itemize}
      \item $\Sigma_i^{\sR} = \trim(\Sigma_{i-1}^{\sC})$, where $\Sigma_0^{\sC}
      = \Sigma$.
      \item $\Sigma_i^{\sC} \in \flexSCC(\Sigma_i^{\sR})$ is a path-flexible
      component in the automaton $\mathcal{M}_{\Pi_i}$ associated with $\Pi_i =
      (\Delta, \Sigma_i^{\sR}, \mathcal{V})$.
    \end{itemize}
    \item When $\Pi = (\Delta, \Sigma, \mathcal{V}, \mathcal{E})$ is a problem
    on unrooted trees, a \emph{good sequence} is a tuple $s = (\mathcal{V}_1,
    \mathcal{D}_1, \dots, \mathcal{V}_k)$ alternating elements of
    $\multich{\Sigma}{\Delta}$ and $\multich{\Sigma}{2}$ where:
    \begin{itemize}
      \item $\mathcal{V}_i$ is obtained by the following procedure:
      \begin{enumerate} 
        \item Restrict the problem $\Pi_{i-1} = (\Delta, \Sigma,
        \mathcal{V}_{i-1}, \mathcal{E})$ to $\mathcal{D}_i$ (where
        $\mathcal{V}_0 = \mathcal{V}$ and $\mathcal{D}_0 = \multich{\Sigma}{2}$)
        to obtain a problem $\Pi_{i-1}' = \rest{\Pi_{i-1}}{\mathcal{D}_i}$ with
        node constraints $\mathcal{V}_{i-1}'$.
        \item Apply trim, that is, set $\mathcal{V}_i =
        \trim(\mathcal{V}_{i-1}')$.
      \end{enumerate}
      \item $\mathcal{D}_i \in \flexSCC(\mathcal{V}_i)$ is a path-flexible
      component in the automaton $\mathcal{M}_{\Pi_i}$ associated with $\Pi_i =
      (\Delta, \Sigma, \mathcal{V}_i, \mathcal{E})$.
    \end{itemize}
  \end{itemize}
  In both cases we refer to $k$ as the \emph{length} of $s$.
\end{definition}

Note that the step of performing $\trim$ (in both contexts) is deterministic,
whereas picking a path-flexible strongly connected component might be a
non-deterministic one.

Finally the depth $d_\Pi$ is defined as the length of the longest good sequence.

\begin{definition}[Depth]%
  \label{def:lcl-depth}
  Let $\Pi$ be an LCL problem on regular trees (of either kind).
  The \emph{depth} $d_{\Pi}$ of $\Pi$ is the largest integer $k$ for which there
  is a good sequence for $\Pi$ of length $k$.
  If no good sequence exists, then $d_\Pi = 0$.
  If there is a good sequence of length $k$ for every $k$, then $d_\Pi=\infty$.
\end{definition}

\section{Super-logarithmic Gap for LCLs on Rooted Regular Trees}
\label{sec:rooted-regular-super-logarithmic}

In this section, we show that if an LCL problem on rooted regular trees has complexity $\omega(\log n)$, then its complexity is $\Omega(n^{1/k})$ for some $k \in \N_+$ in the \rolcl model.
Moreover, this complexity coincides with the existing upper bounds in \congest and \local~\cite{balliu2022efficient}, providing a tight classification.
Formally, we prove the following theorem, restated for the reader's convenience:
\restateThmRootedSuperLogarithmic*

For the rest of this section, we consider $\Pi = (\Delta, \Sigma, \VV)$ to be an
LCL problem on rooted trees with $d_\Pi = k \in \N_+$.
We show that problem~$\Pi$ has locality~$\Omega(n^{1/k})$ by constructing a randomized input instance that no locality-$o(n^{1/k})$ \rolcl algorithm can solve with high probability.
We do this by adapting the previous construction by Balliu et al.~\cite{balliu2022efficient}.
We augment this construction by providing a randomized processing order of nodes.
We then show that the labeling produces by the algorithm for a specific set of nodes must be independent of the exact order.
Finally, we prove by induction that either the algorithm must fail with probability more than~$\frac{1}{n}$, or there exists a good sequence longer than~$d_\Pi$.

\subsection{Lower-Bound Graph Construction}

We start by setting constant~$\beta$ to be the smallest natural number such that for each subset of labels~$\tilde\Sigma \subseteq \Sigma$ and for each label~$\sigma \in \tilde\Sigma \setminus \trim(\tilde\Sigma)$, there exists no correct labeling of the complete tree~$T_\beta$ using only labels in $\tilde\Sigma$ and having root labeled with~$\sigma$.
We also let~$t$ be a constant that we fix later.
We now define the parts of the lower bound graph followed by the main lower bound construction:
\begin{definition}[Lower bound graph part~\cite{balliu2022efficient}]
    \label{def:lb-graphs-2}
    Let $s = 4t + 4$.
    \begin{itemize}
        \item Let $G_{\sR,1}$ be the complete rooted tree~$T_\beta$.
        We say that all nodes of $G_{\sR,1}$ are in layer~$(\sR, 1)$.
        \item For each integer $i \ge 1$, let~$G_{\sC,i}^\circ$ be the tree formed by an $s$-node directed path $v_1 \leftarrow v_2 \leftarrow \dots \leftarrow v_s$ such that we append~$\Delta-1$ copies of~$G_{\sR,i}$ as the children of all nodes~$v_i$.
        We call the path $v_1 \leftarrow v_2 \leftarrow \dots \leftarrow v_s$ the \emph{core path of~$G_{\sC,i}^\circ$}, and say that they are in layer~$(\sC, i)$.
        \item Let $G_{\sC,i}$ be defined like~$G_{\sC,i}^\circ$ except that we append another copy of~$G_{\sR,i}$ as a child of~$v_s$.
        \item For each $i \ge 2$, let~$G_{\sR,i}$ be the tree formed by taking the complete rooted tree~$T_\beta$ and appending~$\Delta$ copies of~$G_{\sC,i-1}$ to each leaf of~$T_\beta$.
        All nodes of~$T_\beta$ are said to be in layer~$(\sR, i)$.
    \end{itemize}
    This construction is visualized in~\cref{fig:Grc}.
\end{definition}

\begin{figure}[tbp]
    \centering
    \includegraphics[width=\textwidth]{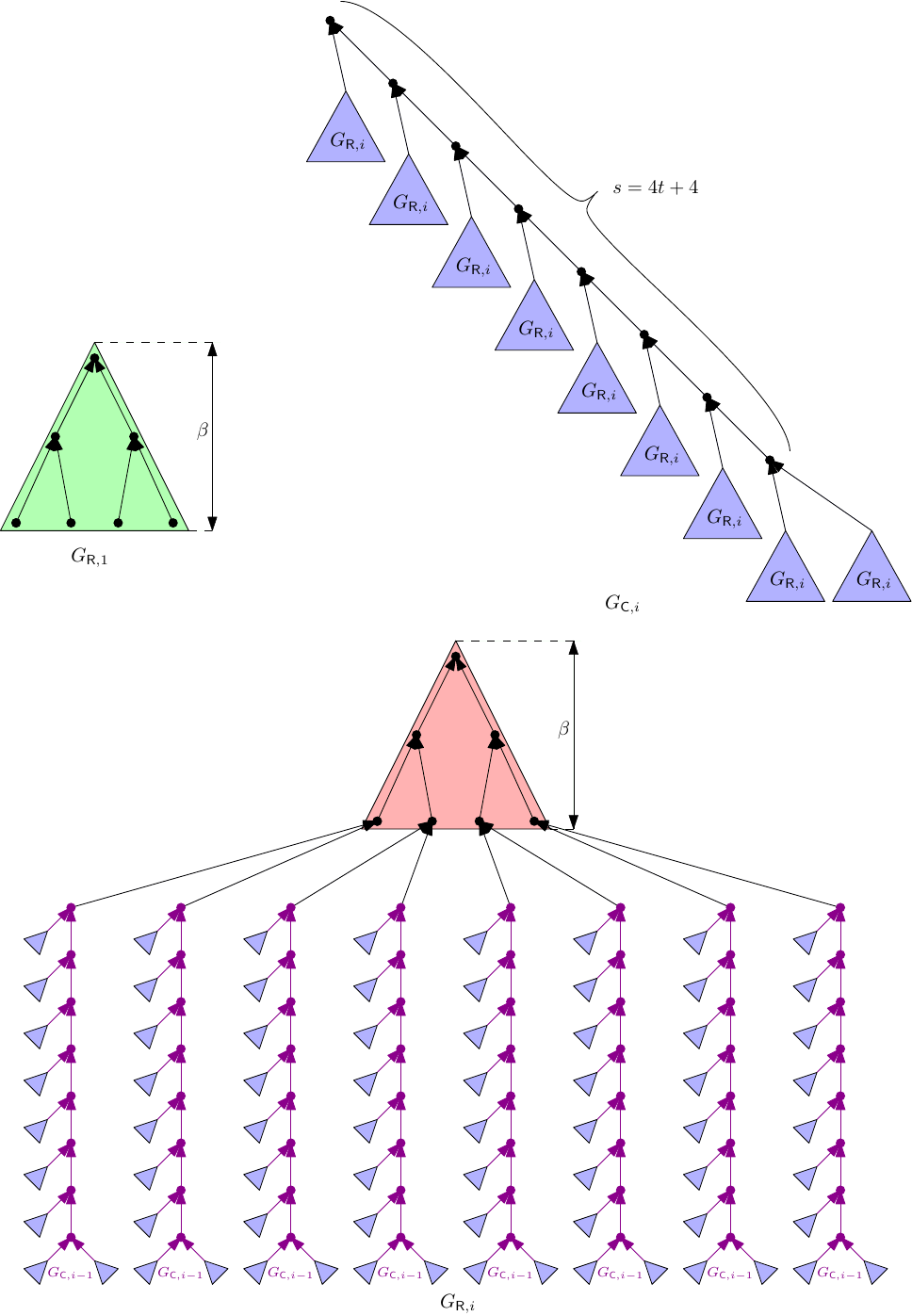}
    \caption{Visualization of lower bound trees; see \cref{def:lb-graphs-2}. The tree on top-left shows the structure of $G_{\sR, 1}$, the tree on top right shows $G_{\sC, i}$ for $i \ge 1$, and the three at the bottom shows $G_{\sR, i}$ for $i \ge 2$.}\label{fig:Grc}
\end{figure}

\begin{definition}[Main lower bound graph~\cite{balliu2022efficient}]
    \label{def:lb-main}
    We define our main lower bound graph~$G$ as the rooted tree formed by the following construction:
    \begin{itemize}
        \item Construct rooted trees $G_{\sC,1}^\circ, G_{\sC,2}^\circ, \dots, G_{\sC,k}^\circ$ and $G_{\sR,k+1}$.
        \item Let $P_i = v_1^i \leftarrow v_2^i \leftarrow \dots \leftarrow v_s^i$ be the core path of $G_{\sC,i}^\circ$, and let $r$ be the root of $G_{\sR,k+1}$.
        \item Add edges $v_s^1 \leftarrow v_1^2, v_s^2 \leftarrow v_1^3, \dots, v_s^{k-1} \leftarrow v_1^k$ and $v_s^k \leftarrow r$.
    \end{itemize}
    This construction is visualized in~\cref{fig:G_lowerbound_uil}.
\end{definition}

The nodes of the main lower bound graph~$G$ are partitioned into layers by the recursive construction.
We order these nodes into the following order:
\begin{equation*}
    (\sR,1) \prec (\sC,1) \prec (\sR,2) \prec (\sC,2) \prec \dots \prec (\sR,k+1)
\end{equation*}
Intuitively, nodes closer to leaves come before nodes closer to the center of the tree.

It is easy to see that each node in~$G$ has degree either~$0$ or~$\Delta$, and that the number of nodes in~$G$ is $n = O(t^k)$.
To show that~$\Pi$ requires locality~$\Omega(n^{1/k})$ to solve, it suffices to show that no \rolcl algorithm can solve~$\Pi$ on~$G$ with locality~$t$ or less.

Next, we define a randomized order~$Z$ in which the adversary reveals nodes of~$G$ to the \rolcl algorithm:
\begin{definition}[Randomized adversarial sequence and nodes $u^i_0, u^i_1, \ldots, u^i_{c_i}$]\label{def:rand-ad-seq}
    A randomized adversarial sequence $Z$ is a random sequence of all nodes in $G$, sampled by the followed process.
    \begin{enumerate}
        \item Start with an empty sequence $Z$.
        \item For each $i \in \{1, \dots, k\}$, do the following:
        \begin{enumerate}
            \item Let $Q^{i}_0, Q^{i}_1, Q^{i}_2, \ldots, Q^{i}_{c_i}$ be all core paths in layer~$(\sC, i)$, where $Q^{i}_0$ is $P_i$, that is the core path closest to the root.
            \item For each $l \in \{1, \dots, c_i\}$, do the following:
            \begin{enumerate}
                \item Let $r^i_l$ be the node in $Q^i_l$ that is closest to the root of $G$.
                \item Sample $d^i_l$ from the set $\{2t+1, 2t+2\}$ uniformly at random.
                \item  Let $u^i_l$ be the node of $Q^i_l$ that is at the distance of $d^i_l$ from $r^i_l$.
            \end{enumerate}
            \item Append a uniform random permutation of $u^{i}_0, u^i_1, \ldots, u^i_{c_i}$ to $Z$.
        \end{enumerate}
        \item Append all other nodes in any fixed order to $Z$.
    \end{enumerate}
    The randomized adversarial sequence $Z$ therefore also defines the nodes $u^i_0, u^i_1, \ldots, u^i_{c_i}$.  
\end{definition}

\begin{figure}
    \centering
    \includegraphics[width=0.96\textwidth]{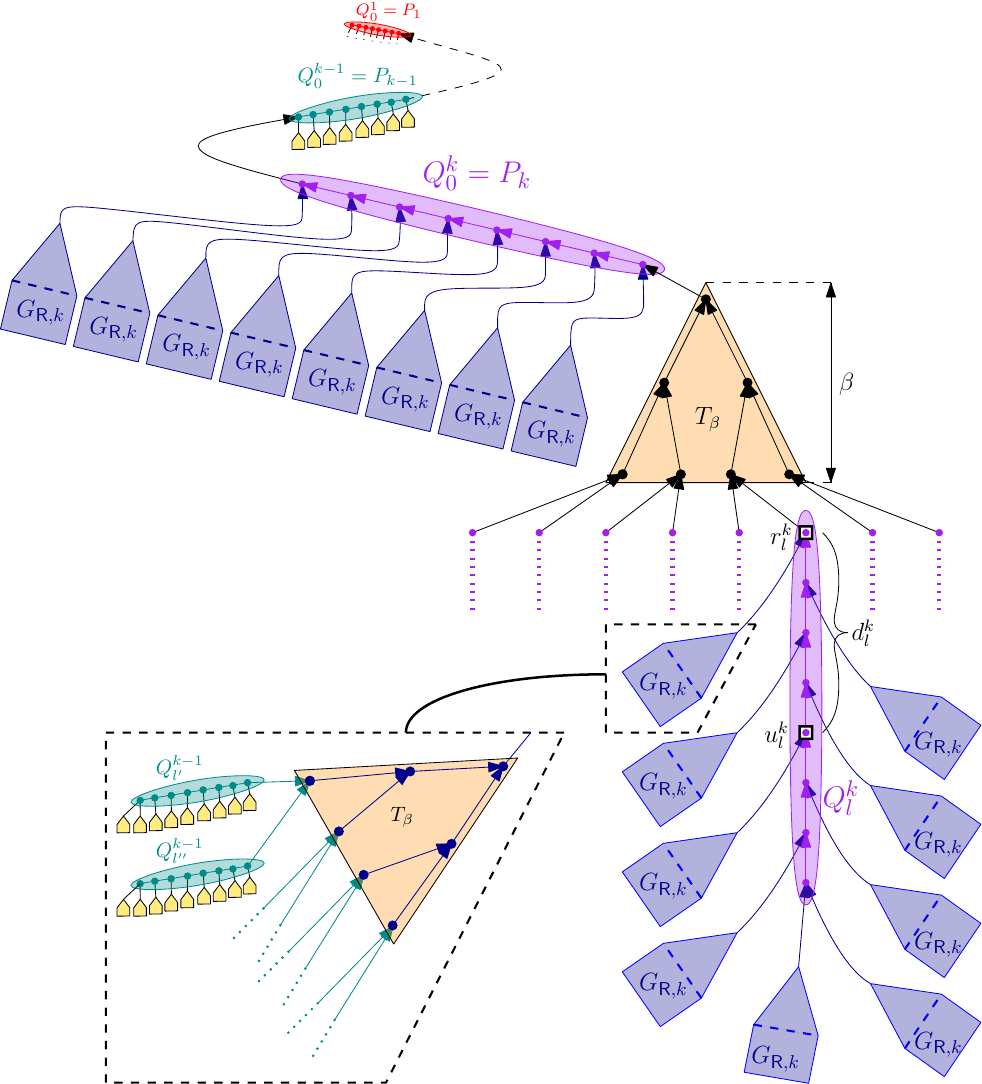}
    \caption{The figure shows the lower bound graph $G$. The orange triangles represent complete trees $T_\beta$. Each purple oval represents a core path $Q^k_{l}$ of some $G_{\sC, k}$. Each blue pentagon represents copies of $G_{\sR,k}$. The zoom-in shows the structure of the internal connections of $G_{\sR,k}$. The turquoise chains are core paths $Q^{k-1}_{l'}$ of $G_{\sC,k-1}$ where each yellow pentagon represents copies of $G_{\sR, k-1}$. The topmost red oval represents the core path $P_1$ (which is same as $Q^1_0$) of the topmost copy of $G_{\sC, 1}$. Near the zoom-in, the figure shows how $u^k_l$ is chosen in $Q^k_l$ from $r^k_l$ and $d^k_l$. See \cref{def:lb-main,def:rand-ad-seq} for more information.}\label{fig:G_lowerbound_uil}
\end{figure}

Note that for either choice of $d^i_l \in \{2t+1,2t+2\}$, the distance of $u^i_l$ from either end of the core path it belongs to is at least $2t + 1$.
This is the reason behind the choice of $s = 4t + 4$.
In particular, this ensures that any locality-$t$ \rolcl algorithm~$\algo$ must label nodes~$u^i_l$ for $i \in \{1, 2, \ldots, k\}$ and $l \in \{0, 1, \ldots, c_i\}$ independently of the adversarial sequence~$Z$:
\begin{lemma}[Independence lemma]\label{lem:independent-label}
    For $i \in \{1, 2, \ldots k\}$, the labelling of the nodes $u^i_0, u^i_1, \ldots, u^i_{c_i}$ produced by any \rolcl algorithm~$\algo$ with locality $t$ is independent of the randomness of the randomized adversarial sequence~$Z$.
\end{lemma}
\begin{proof}
    For any $i$, the $t$-neighborhoods of $u^i_0, u^i_1, \ldots, u^i_{c_i}$ are identical and disjoint.
    For $i_1 \ne i_2$, the $t$-neighborhoods of $u^{i_1}_{j_1}$ and $u^{i_2}_{j_2}$ are also disjoint as their lowest common ancestor is at a distance of at least $2t$ from one of them.
    Therefore, for any $u^i_l$, its $t$-neighborhood is unlabeled for all randomized adversarial sequences $Z$.
    Since the $t$-neighborhoods of $u^i_0, u^i_1, \ldots, u^i_{c_i}$ are the identical, disjoint and unlabeled at the time of revealing, the output distribution of $u^i_0, u^i_1, \ldots, u^i_{c_i}$ produced by $\mathcal A$ is independent of the randomized adversarial sequence $Z$ for all $i \in \{1, 2, \ldots k\}$.
\end{proof}

Next, we define a sequence of subsets of the nodes in $G$.
This definition is very similar to the subsets defined in previous works of Balliu et al.~\cite{balliu2022efficient}.

\begin{definition}[Subsets of nodes in $G$]\label{def:subsets}
    We define the following subsets of nodes in $G$:
    \begin{itemize}
       \item Define $S'_{\sR,1}$ as the set of nodes $v$ in $G$ such that the subgraph induced by $v$ and its descendants within radius-$\beta$ is isomorphic to $T_\beta$.
        \item For $2 \leq i \leq k+1$, define $S'_{\sR,i}$ as the set of nodes $v$ in $G$ such that the subgraph induced by $v$ and its descendants within radius-$\beta$ is isomorphic to $T_\beta$ and contains only nodes in set~$S'_{\sC,i-1}$.
        \item For $1 \leq i \leq k$, define $S'_{\sC,i}$ as the set of nodes $v$ in $G$ meeting one of the following conditions.
        \begin{itemize}
            \item $v$ is in one of these layers: $(\sR,i+1), (\sC, i+1), (\sR, i+2), \ldots, (\sR, k+1)$.
            \item $v \in P_i$ is either the node $u^i_0$ or its descendant in layer $(\sC,i)$.
            \item $v \notin P_i$ is either the node $u^i_l$ or an ancestor of it in layer $(\sC,i)$, for some $l \ne 0$.      
        \end{itemize}
    \end{itemize}
\end{definition}

The intuition behind the definition is the following:
$S'_{\sR, 1}$ basically consists of all nodes of the tree $G$ apart from some layer $(\sR, 1)$ nodes which are very close to leaves.
$S'_{\sC, i}$ contains the nodes which are descendants of $u^i_0$ but also an ancestor of some node $u^i_l$.
It is evident from the definition that $S'_{\sC, i}$ contains nodes only from layer $(\sC, i)$ or above.
And finally, $S'_{\sR, i}$ contains of only nodes in layer $(\sR, i)$ or above.
In particular, it contains the nodes, and their immediate children, that are roots of subgraphs isomorphic to~$T_\beta$ that are fully in layer~$(\sR, i)$.

We prove some basic properties of the sets in \cref{def:subsets}. 
The proof is again similar to that in the previous works of Balliu et al.~\cite{balliu2022efficient}, we restate it for the convenience of the reader:
\begin{lemma}[Subset containment]\label{lem:containment}
    We have $S'_{\sR,1} \supseteq S'_{\sC,1} \supseteq \cdots \supseteq S'_{\sR,k+1} \neq \emptyset$.
\end{lemma}
\begin{proof}
    The claim that $S'_{\sC,i} \supseteq S'_{\sR,i+1}$ follows from the definition of $S'_{\sR,i+1}$ that $v \in S'_{\sC,i}$ is a necessary condition for $v \in S'_{\sR,i+1}$.
    To prove the claim that $S'_{\sR,i} \supseteq S'_{\sC,i}$, we recall that $v \in S'_{\sC,i}$ implies that  $v$ is in layer $(\sC,i)$ or above. By the construction of $G$, the subgraph induced by $v$ and its descendants within the radius-$\beta$ neighborhood of $v$ is isomorphic to $T_\beta$ and contains only nodes in layer $(\sR,i)$ or above.  
    Since all nodes in layer $(\sR,i)$ or above are in $S'_{\sC,i-1}$, we infer that $v \in S'_{\sC,i}$ implies $v \in S'_{\sR,i}$. 

    To see that $S'_{\sR,k+1} \neq \emptyset$, consider the node $r$, which is the root of $G_{\sR,k+1}$.
    The subgraph induced by $r$ and its descendants within radius-$\beta$ neighborhood of $r$ is isomorphic to $T_\beta$ and contains only nodes in layer $(\sR,k+1)$. We know that all nodes in layer $(\sR,k+1)$ are in $S'_{\sC,k}$, so $r \in S'_{\sR,k+1}$.
\end{proof}

It is also evident that the children of the nodes in $S'_{\sC, i}$ are also contained in $S'_{\sR, i}$, for all $i \in \{1, 2, \ldots, k\}$:
\begin{observation}\label{obs:children-containment}
    If $v$ is a child of some node in $S'_{\sC, i}$, then $v \in S'_{\sR, i}$, for all $i \in \{1, 2, \ldots, k\}$.
\end{observation}
Furthermore, due to the recursive structure of the input instance $G$, we get the following result:
\begin{observation}\label{obs:path-from-ci-to-ci}
    For each $i \in \{1, 2, \ldots, k\}$, for all $v \in S'_{\sC, i}$ there exists a path $P$ from $u^i_l$ (for some $l$) to $u^i_0$, such that $P$ contains $v$.
    Moreover, all nodes of $P$, and their children, belong to $S'_{\sR, i}$.
\end{observation}

\subsection{Running Algorithm on the Lower-Bound Graph}

We are now ready to show that no locality-$t$ \rolcl algorithm can solve problem~$\Pi$ on graph~$G$ and processing order~$Z$ with high probability.
For contradiction, we assume that such algorithm~$\algo$ exists.
We then show that algorithm~$\algo$ fails to find a proper labeling for lower bound graph~$G$ with randomized input sequence~$Z$ with probability larger than~$\frac{1}{n}$.
In particular, this shows that $\algo$ does not work with high probability.
Hence the locality any \rolcl algorithm solving~$\Pi$ must be strictly larger than $t = \Omega(n^{1/k})$.

We start by proving the following result that holds for every valid output labelling:
\begin{lemma}\label{lem:restr-probl-path}
    Let $L$ be a valid labeling of the nodes according to~$\Pi$.
    For each $i \in \{1, 2, \ldots, k\}$, let $\tilde{\Sigma}^L_i \subseteq \Sigma$ be any subset of labels such that for each node $v \in S'_{\sR, i}$, the corresponding label $L(v)$ is in $\tilde{\Sigma}^L_i$.
    Let $\mathcal M$ be the automaton associated with the restricted problem
    $\rest{\Pi}{\tilde{\Sigma}^L_i}$.  
    Then for all nodes $u', u'' \in S'_{\sC, i}$ such that $u' \leftarrow u''$
    is an edge in $G$, edge $L(u') \rightarrow L(u'')$ exists in~$\mathcal M$.
\end{lemma}
Note that the result talks about existence of the edge $L(u') \rightarrow L(u'')$ in the automaton of the restricted problem $\rest{\Pi}{\tilde{\Sigma}^L_i}$.
The result would trivially hold for the automaton of the original problem $\Pi$.
\begin{proof}
    By \cref{obs:path-from-ci-to-ci}, there exists path $P = (u^i_0 = v_0 \leftarrow v_1 \leftarrow \cdots \leftarrow v_\alpha = u^i_l)$ for some $l$ containing node~$u''$, and hence also~$u'$.
    Moreover, all nodes of $P$ and their children belong to set $S'_{\sR, i}$.
    This implies that the labels of all nodes of $P$, and labels of their children, belong to set~$\tilde{\Sigma}^L_i$.

    Let $\sigma \in \tilde{\Sigma}^L_i$ be the label of $u'$, and let $\sigma_1, \sigma_2, \dots, \sigma_\Delta$ be the labels of its children.
    As the children of $u'$ belong to $S'_{\sR, i}$, all $\sigma_j$ belong to $\tilde{\Sigma}^L_i$.
    Therefore $(\sigma : \sigma_1, \dots, \sigma_\Delta)$ must be a valid configuration in the restricted problem~$\rest{\Pi}{\tilde{\Sigma}^L_i}$.
    This then implies that all $\sigma \rightarrow \sigma_1, \dots, \sigma \rightarrow \sigma_\Delta$ are edges in~$\mathcal M$.
    As the label of node~$u''$ is one of $\sigma_1, \dots, \sigma_\Delta$, edge $\sigma = L(u') \rightarrow L(u'') = \sigma_j$ belongs to~$\mathcal M$.
\end{proof}

We now restrict our attention to a particular run of algorithm~$\algo$ on the randomized sequence~$Z$ on graph~$G$.
We start by defining the set of labels the algorithm uses to label nodes $u^i_0, u^i_1, u^i_2, \ldots, u^i_{c_i}$ (see \cref{def:rand-ad-seq}):
\begin{definition}\label{def:Gamma-prime}
    For $i \in \{1, 2, \ldots k\}$, let $\Sigma'_i \subseteq \Sigma$ be the set of labels used by $\mathcal A$ to label the nodes $u^i_0, u^i_1, u^i_2, \ldots, u^i_{c_i}$.
\end{definition}
Note that the set $\Sigma'_i$, $i \in \{1, 2, \ldots, k\}$ depends on a particular run of $\mathcal A$, and can change with different runs of $\mathcal A$.
Whenever we talk about $\Sigma'_i$, $i \in \{1, 2, \ldots, k\}$, we consider a fixed run of $\mathcal A$ and analyze it.

We now prove the following result on the structure of these sets of labels:
\begin{lemma}\label{lem:inductive-res}
    Let $j$ be an integer in $\{0, 1, 2, \ldots, k\}$, and let $\Sigma'_1, \Sigma'_2, \ldots, \Sigma'_k$ be the sets of nodes $\algo$ used to label nodes~$u^i_0, u^i_1, u^i_2, \ldots, u^i_{c_i}$.
    If there now exists sets~$\Sigma = \Sigma''_0, \Sigma''_1, \Sigma''_2, \ldots, \Sigma''_j$ satisfying $\Sigma''_i \in \flexSCC(\trim(\Sigma''_{i-1}))$ and $\Sigma'_i \subseteq \Sigma''_i$ for all  $i \in \{1, 2, \ldots j\}$, then the following must hold for the full labeling~$L$ produced by~$\algo$:
    \begin{itemize} 
        \item \textbf{Induction hypothesis for $S'_{\sR, i}$:} for all $i \in \{1, 2, \ldots, j, j+1\}$ and $v \in S'_{\sR, i}$ we have $L(v) \in \trim(\Sigma''_{i-1})$.
        \item \textbf{Induction hypothesis for $S'_{\sC, i}$:} for all $i \in \{1, 2, \ldots, j\}$ and $v \in S'_{\sC, i}$ we have $L(v) \in \Sigma''_i$.
    \end{itemize}
\end{lemma}
The idea behind defining $\Sigma''_0, \Sigma''_1, \ldots, \Sigma''_j$ is to ensure that $(\SigmaR{1} = \trim(\Sigma''_0), \SigmaC{1} = \Sigma''_1, \SigmaR{2} = \trim(\Sigma''_1),\ldots, \SigmaC{j} = \Sigma''_j, \SigmaR{j+1} = \trim(\Sigma''_{j-1}))$ forms a good sequence and the corresponding nodes follow some specific structure.

\begin{proof}
    We prove this lemma with induction on $i$ for a fixed $j \le k$:

    \subparagraph*{Base Case: induction hypothesis for $S'_{\sR,1}$.}
    By definition, we have $\Sigma''_0 = \Sigma$.
    Moreover, by definition, for every node $v \in S'_{\sR, 1}$, the radius-$\beta$ descendants of $v$ are isomorphic to $T_\beta$.
    For any valid assignments, the labels of $v$ and its radius-$\beta$ descendants come from the set $\Sigma = \Sigma''_0$.
    By definition of $\beta$, this must mean that the label of $v$ must come from $\trim(\Sigma) = \trim(\Sigma''_0)$.
    Therefore, for all $v \in S'_{\sR, i}$, we have $L(v) \in \trim(\Sigma''_{i-1})$ for $i=1$.
    
    \subparagraph*{Induction: induction hypothesis for $S'_{\sC, i}$.}
    We assume that for some $1 \le i \le j$, the induction hypothesis for $S'_{\sR, i}$ holds true, i.e. for all $v \in S'_{\sR, i}$ we have $L(v) \in \trim(\Sigma''_{i-1})$.
    Let $v$ be any node in $S'_{\sC, i}$.
    By \cref{obs:path-from-ci-to-ci}, there exists a path $P$ from $u^i_l$ for some $l$ to $u^i_0$ containing $v$ where all nodes of $P$ and the children of all such nodes belong to $S'_{\sR, i}$.
    Let the path $P$ be $u^i_0 = v_0 \leftarrow v_1 \leftarrow \cdots \leftarrow v_\alpha = u^i_l$.
    Then $v_y$ and its children are in $S'_{\sR, i}$, for all $y \in \{0, 1, \ldots, \alpha\}$.
    Moreover, by \cref{lem:restr-probl-path}, for all $y \in \{0, 1, \ldots, \alpha-1\}$, $L(v_y) \rightarrow L(v_{y+1})$ is a directed edge in $\mathcal M$, where $\mathcal M$ is the automaton associated with the restricted problem $\rest{\Pi}{\trim(\Sigma''_{i-1})}$.

    Now, $L(u^i_0), L(u^i_l) \in \Sigma'_i$ and $\Sigma'_i \subseteq \Sigma''_i$ where $\Sigma''_i$ is a flexible strongly connected component in $\mathcal M$ by the definition of~$\Sigma''_i$.
    Since $\mathcal M$ has the directed edge $L(v_y) \rightarrow L(v_{y-1})$ for all $y \in \{0, 1, \ldots, \alpha -1\}$, $L(u^i_0) = L(v_0) \rightarrow L(v_1) \rightarrow \cdots \rightarrow L(v) \rightarrow \cdots \rightarrow L(v_\alpha) = L(u^i_l)$ is a walk in $\mathcal M$ that starts and ends in the strongly connected component $\Sigma''_i$.
    This must mean $L(v) \in \Sigma''_i$, for all $v \in S'_{\sC, i}$.

    \subparagraph*{Induction: induction hypothesis for $S'_{\sR, i}$.}
    We assume that for some $2 \le i \le j+1$, the induction hypothesis for $S'_{\sC, i-1}$ holds true, that is for all $v \in S'_{\sC, i-1}$ we have $L(v) \in \Sigma''_{i-1}$.
    Moreover, by definition, each node~$v \in S'_{\sR, i}$ is a root of a subtree isomorphic to~$T_\beta$ such that all nodes reside in~$S'_{\sC, i-1}$, each of which has a label from set~$\Sigma''_{i-1}$ by induction hypothesis for~$S'_{\sC, i-1}$.
    Hence, the label of $v$ must be from $\trim(\Sigma''_{i-1})$, implying for all $v \in S'_{\sR, i}$, we have $L(v) \in \trim(\Sigma''_{i-1})$.
\end{proof}

We can now combine this result with the independence lemma (\cref{lem:independent-label}) to provide a lower bound for the failure probability of~$\algo$ on the randomized adversarial sequence~$Z$, even after fixing the labels produced by~$\algo$ on nodes~$u^i_0, u^i_1, \ldots, u^i_{c_i}$ for all $i \in \{1, 2, \ldots, k\}$:
\begin{lemma}\label{lem:lowerbound}
    For all choices of $\Sigma'_1, \Sigma'_2, \ldots, \Sigma'_k$, algorithm~$\algo$ fails to produce a correct labelling for the randomized adversarial sequence $Z$ of $G$ with probability more than $\frac{1}{n}$.
\end{lemma}

\begin{proof}
    Let $j$ be the largest integer in $\{0, 1, 2, \ldots, k\}$ such that there exists $\Sigma''_0, \Sigma''_1, \Sigma''_2, \ldots, \Sigma''_j$ satisfying the conditions of \cref{lem:inductive-res}, that is:
    \begin{itemize}
        \item $\Sigma''_0 = \Sigma$,
        \item for all $i \in \{1, 2, \ldots j\}, \Sigma''_i \in \flexSCC(\trim(\Sigma''_{i-1}))$ and $\Sigma'_i \subseteq \Sigma''_i$.
    \end{itemize}
    Note that such $j$ always exists as these conditions trivially hold true for $j=0$.

    We now do a case analysis on~$j$ and show that in each case the algorithm fails to produce a correct labeling with probability more than~$\frac{1}{n}$:
    \subparagraph*{\boldmath Case I: $j < k$.}
    This means that there does not exist $\Sigma''_{j+1} \supseteq \Sigma'_{j+1}$ such that $\Sigma''_{j+1} \in \flexSCC(\trim(\Sigma''_j))$.
    Let $\mathcal M$ be the automaton associated with the restricted problem
    $\rest{\Pi}{\trim(\Sigma''_j)}$.
    Now one of the following must hold:
    \begin{enumerate}[(a)]
        \item $\Sigma'_{j+1} \nsubseteq \trim(\Sigma''_{j})$,
        \item $\Sigma'_{j+1} \subseteq \trim(\Sigma''_{j})$ but $\Sigma'_{j+1}$ does not form a strongly connected component of $\mathcal M$,
        \item $\Sigma'_{j+1} \subseteq \trim(\Sigma''_{j})$ and $\Sigma'_{j+1}$ forms an inflexible strongly connected component of $\mathcal M$.
    \end{enumerate}
    Observe that if all three are false, then we could take the strongly connected component containing the elements of $\Sigma'_{j+1}$ in $\mathcal M$ as $\Sigma''_{j+1}$.
    This would contradict our assumption that there does not exist $\Sigma''_{j+1} \in \flexSCC(\trim(\Sigma''_{j}))$.

    We now show that for all three cases, the probability that $\mathcal A$ eventually fails to output a correct solution is more than $\frac{1}{n}$:
    \begin{description}
        \item[Case I(a): $\Sigma'_{j+1} \nsubseteq \trim(\Sigma''_{j})$.]
        There must be some $l$ such that label~$L(u^{j+1}_l)$ generated by $\algo$ for node~$u^{j+1}_l$ is not in set $\trim(\Sigma''_j)$.
        By \cref{lem:inductive-res}, we know that in any valid labelling we have $L(v) \in \trim(\Sigma''_{j})$ for every node~$v \in S'_{\sR,j+1}$.
        However, this is violated in this particular run as $u^{j+1}_l \in S'_{\sC,j+1} \subseteq S'_{\sR,j+1}$ but $L(u^{j+1}_l) \notin \trim(\Sigma''_j)$.
        Therefore, irrespective of the other output labels, $\algo$ must fail to generate a valid labelling for $G$.
        Thus, the probability that $\mathcal A$ eventually fails is indeed $1$.

        \item[Case I(b): $\Sigma'_{j+1}$ does not from a strongly connected component of $\mathcal M$.]
        There must be two labels $\sigma_1, \sigma_2 \in \Sigma'_{j+1}$ such that $\sigma_1$ and $\sigma_2$ are not in the same strongly connected component in $\mathcal M$.
        This means that at least one of walks $\sigma_1 \rightsquigarrow \sigma_2$ and $\sigma_2 \rightsquigarrow \sigma_1$ does not exist in~$\mathcal M$.
        Without loss of generality, we may assume that walk~$\sigma_1 \rightsquigarrow \sigma_2$ is missing from~$\mathcal M$.
        Algorithm~$\algo$ uses label~$\sigma_1$ to label at least one of nodes~$u^{j+1}_0, u^{j+1}_1, \ldots, u^{j+1}_{c_{j+1}}$.
        As the choices of the algorithm are independent of the input sequence~$Z$ by \cref{lem:independent-label}, and the sequence~$Z$ contains a uniformly random permutation of the nodes~$u^{j+1}_0, u^{j+1}_1, \ldots, u^{j+1}_{c_{j+1}}$, the probability that~$\algo$ uses label~$\sigma_1$ on node~$u^{j+1}_0$ is at least $\frac{1}{c_{j+1}+1} > \frac{1}{n}$.

        We now show that algorithm~$\algo$ is bound to fail in this case.
        As label~$\sigma_2 \in \Sigma'_{j+1}$, algorithm~$\algo$ must use it to label at least one of the nodes~$u^{j+1}_1, \ldots, u^{j+1}_{c_{j+1}}$.
        Let that node be~$u^{j+1}_l$.
        Consider path~$P = u^{j+1}_0 \leftarrow v_1 \leftarrow \dots \leftarrow v_\alpha \leftarrow u^{j+1}_l$.
        Note that all nodes of~$P$, and their children, belong to set~$S'_{\sR, j}$.
        \Cref{lem:restr-probl-path} now gives us that walk~$\sigma_1 = L(u^{j+1}_0) \rightarrow L(v_1) \rightarrow \dots \rightarrow L(v_\alpha) \rightarrow L(u^{j+1}_l) = \sigma_2$ exists in~$\mathcal M$, which contradicts our assumption that $\sigma_1$ and $\sigma_2$ don't belong to the same strongly connect component.
        
        As this happens with probability at least $\frac{1}{c_{j+1}+1} > \frac{1}{n}$, algorithm~$\algo$ must fail with at least the same probability.

        \item[Case I(c): $\Sigma'_{j+1}$ forms an inflexible strongly connected component of $\mathcal M$.]\hskip 0pt plus 10pt
        Similar to the previous case, there must exist a walk $L(u^{j+1}_0) \rightsquigarrow L(u^{j+1}_l)$ in $\mathcal M$ for each $l \in \{1, 2, \allowbreak\ldots, c_{j+1}\}$.
        For any pair of labels~$(L(u^{j+1}_0), L(u^{j+1}_l))$ which are in an inflexible strongly connected component in $\mathcal M$, automaton~$\mathcal M$ cannot contain walks~$L(u^{j+1}_0) \rightsquigarrow L(u^{j+1}_l)$ of length both $d$ and $d+1$; otherwise the component containing the pair would be flexible.
        However, the distance between nodes $u^{j+1}_0$ and $u^{j+1}_l$ is even with probability~$\frac{1}{2}$ and odd with probability~$\frac{1}{2}$ by construction of sequence~$Z$.
        As the labels chosen by~$\algo$ for nodes $u^{j+1}_0$ and $u^{j+1}_l$ are independent of sequence~$Z$ by \cref{lem:independent-label}, it must fail to choose a label with the correct parity with probability at least~$\frac{1}{2}$.
        Hence the algorithm must fail with probability at least $\frac{1}{2} > \frac{1}{n}$.
    \end{description}
    This proves that if $j < k$, the probability that the algorithm eventually fails is more than $\frac{1}{n}$.
    The other case $j=k$ is relatively simpler:
    \subparagraph*{\boldmath Case II: $j=k$.}
    Let $v$ be a node in $S'_{\sR, k+1}$; note that $S'_{\sR, k+1} \ne \emptyset$ by \cref{lem:containment}.
    Due to \cref{lem:inductive-res}, the label of node~$v$ must be in $\trim(\Sigma''_k)$ in any correct labelling of the graph, which implies that $\trim(\Sigma''_k) \ne \emptyset$. 
    However, this cannot be the case since
    \begin{equation*}
        \left(
            \trim(\Sigma''_0),
            \Sigma''_1,
            \trim(\Sigma''_1),
            \dots,
            \trim(\Sigma''_{k-1}),
            \Sigma''_k,
            \trim(\Sigma''_k)
        \right)
    \end{equation*}
    would be a good sequence of length $k+1$, contradicting $d_\pi = k$.
    Therefore, algorithm $\mathcal A$ must eventually fail (i.e. with probability $1$).

    Since the algorithm $\mathcal A$ fails with probability more than $\frac{1}{n}$ in all cases, the failure probability of $\mathcal A$ in total must be more than $\frac{1}{n}$.
\end{proof}

The main result now follows as a simple corollary of the previous lemma:
\begin{proof}[Proof of \cref{thm:rooted-full}]
    Let~$\algo$ be a \rolcl algorithm solving LCL problem~$\Pi$ on rooted regular trees having depth~$d_\Pi = k \in \N_+$ with locality~$t = o(n^{1/k})$.
    Then by \cref{lem:lowerbound} algorithm~$\algo$ must fail with probability more than~$\frac{1}{n}$.
    Hence any \rolcl algorithm solving~$\Pi$ with high probability must have locality~$\Omega(n^{1/k})$.
\end{proof}

\section{Lower Bounds in the Super-Logarithmic Region in Unrooted Regular Trees}\label{sec:unrooted-regular}

In this section, we restate and prove the following result:
\restateThmUnrootedFull*

Let us thus fix such a problem $\Pi$ and its depth $k$.
We show that the existence of an $o(n^{1/k})$ algorithm which solves $\Pi$
will force a contradiction.
The procedure is analogous to the rooted case.

We pick $\gamma$ to be the smallest integer satisfying: for each set $S\subseteq \VV$ and each $C \in S\setminus \trim(S)$, there exists no correct labeling of $T_{\gamma}^*$ 
where the node configuration of the root node is $C$ and the node configurations of the remaining $\Delta$-degree nodes are in $S$.
By definition of trim, such a $\gamma$ will always exist. We note that in any algorithm, $\gamma$ is a constant, depending only on the underlying LCL problem $\Pi$.

\begin{definition}[Lower bound graphs]\label{unrooted-reg-LB-graphs}
    Let $t$ be any positive integer, and choose $s = 4t+4$. 
    \begin{itemize}
        \item $G_{\sR,1}$ is the rooted tree $T_{\gamma}$, and $G_{\sR,1}^*$ is the rooted tree $T_{\gamma}^*$. We say all nodes in $G_{\sR,1}$ or $G^*_{\sR, 1}$ are in layer $(\sR, 1)$.
        \item For each integer $i \ge 1$, $G_{\sC, i}$ is constructed as follows. Begin with an $s$-node path $(v_1, v_2, \dots, v_s)$, and let $v_1$ be the root.
        For each $1 \le i< s$, append $\Delta - 2$ copies of $G_{\sR,i}$ to $v_i$. For $i=s$, append $\Delta-1$ copies.
        We say nodes $v_1, v_2, \dots, v_s$ are in layer $(\sC, i)$.
        \item For each integer $i \ge 2$, $G_{\sR, i}$ is constructed as follows. Begin with a rooted tree $T_{\gamma}$. Append $\Delta - 1$ copies of $G_{\sC, i-1}$ to each leaf in $T_{\gamma}$.
        We say all nodes in $T_{\gamma}$ are in layer $(\sR, i)$.
        The rooted tree $G^*_{\sR, i}$ is defined analogously by replacing $T_{\gamma}$ with $T^*_{\gamma}$ in the construction.
    \end{itemize}
\end{definition}

The choice of $s = 4t+4$ will be motivated later.
Note that although the trees $G_{\sR, i}, G^*_{\sR, i},$ and $G_{\sC, i}$ are defined as rooted trees, we can treat them as unrooted trees. We define our main lower bound graph to be $G = G^*_{\sR, k+1}$.
We will use this graph to show that if $d_{\Pi} = k$, there can be no algorithm faster than $\Omega(n^{1/k})$ solving $\Pi$ on $G$.

From the construction, it can easily be seen that $G=G^*_{\sR, k+1}$ has $O(t^k)$ nodes. So to show that $\Pi$ requires a locality of $\Omega(n^{1/k})$, it suffices to show that solving $\Pi$ requires locality of at least $t$ on $G$.

We now define a sequence of subsets of nodes in $G$.

\begin{definition}[Subsets of nodes in $G$: $S_{\sR, i}, S_{\sC, i}, U_i$]
    We define the following sets:
    \begin{itemize}
        \item Let $S_{\sR, 1}$ consist of all nodes in $G$ with their radius-$\gamma$ neighborhood isomorphic to $T_{\gamma}$.
        \item For each path $(v_1, v_2, \dots, v_s)$ in layer $(\sC, i)$, choose one element of the set $\{v_{2t+1}, v_{2t+2}\}$ uniformly at random. Define the set of all such elements as $U_i$.
        \item For $1 \le i \le k$, construct $S_{\sC, i}$ as follows: 
        Initialize $S_{\sC, i}$ as $U_i$. For any $u_j, u_k\in U_i$, there exists a path joining them, $P = (u_j, s_1, s_2, \dots, u_k)$. Add each node $s_{l}$ in such a path to $S_{c, i}$.
        \item For $2 \le i \le k+1$, let $S_{\sR, i}$ consist of all nodes in $S_{\sC, i-1}$ with radius-$\gamma$ neighborhood contained entirely within $S_{\sC, i-1}$ and isomorphic to $T_{\gamma}$.
    \end{itemize}
    
\end{definition}

Intuitively, $S_{\sC, i}$ consists of all nodes in $U_i$ and all nodes ``above'' any node in $U_i$.
The layers $(\sR, i)$ and $(\sC, i)$ form a partition of nodes in $G$; we impose an order on these layers for subsequent discussion: let $(\sR, 1)\preceq (\sC, 1)\preceq \dotsb \preceq (\sR, k) \preceq (\sC, k) \preceq (\sC, k+1)$.
For example, ``layer $(\sR, i)$ or higher'' refers to layers $(\sR, i), (\sC, i), \dots, (\sR, k+1)$. 

We now prove some properties of the sets $S_{\sR, i}$ and $S_{\sC, i}$.

\begin{lemma}[Subset containment]
    It holds that $S_{\sR, 1}\supseteq S_{\sC, 1} \supseteq S_{\sR, 2} \supseteq \dotsb \supseteq S_{\sR, k+1}$.  
\end{lemma}
\begin{proof}
    Each $S_{\sC, i}\supseteq S_{\sR, i+1}$ follows immediately from the definition, as we require each node in $S_{\sR, i+1}$ to be in $S_{\sC, i}$ in the construction.

    To see that $S_{\sR, i} \supseteq S_{\sC, i}$, observe that $S_{\sR, i}$ contains all nodes in layers $(\sR, i)$ and higher, for all $i \ge 2$:
    each node in layer $(\sR, i)$ or higher is a distance at least $2t+1 \gg \gamma$ from any node in $U_i$, and therefore its radius-$\tau$ neighborhood is contained entirely within $S_{\sC, i}$.
    It can easily be seen from the construction that this neighborhood is isomorphic to $T_{\gamma}$. So we have that $S_{\sR, i} \supseteq S_{\sC, i}$.
\end{proof}

We now define our randomized adversarial order. Given a \rolcl algorithm $\algo$ which has locality $t$,
we will show that if $\algo$ is given nodes in this order, it will fail with probability asymptotically greater than $\frac{1}{n}$.

\begin{definition}[Randomized adversarial order]\label{unrooted-rand-adversarial-order}
    We define a \emph{randomized adversarial order} of nodes as follows: begin with an empty sequence $\mathcal{S}$.
    \begin{enumerate}
        \item For $i$ from $1$ to $k$, append the nodes of $U_i$ to $\seq$ in any randomized order.
        \item Add all the remaining nodes (those in $V(G)\setminus \bigcup_i U_i$) to $\seq$ in any order.
    \end{enumerate}
    Also, define $\seq^U$ to be the subsequence of $\seq$ consisting only of nodes in some $U_i$. Note that $\seq^U$ is a prefix of $\seq$. 
\end{definition}

Now, we argue that when labeling the nodes of $\bigcup_i U_i$, the label which $\algo$ assigns the $i$th node of $\seq^U$ is independent of the randomness in $\seq^U$'s order.

\begin{lemma}[Independence of labels of nodes in each $U_i$]\label{independence-lemma-unrooted}
    For $i \in \{1, 2, \dots, k\}$, the labels assigned by $\algo$ to the nodes of $U_i$ are independent of the way the nodes of $U_i$ are ordered in~$\seq^U$.
\end{lemma}
\begin{proof}
    Let $U_i = \{u_1^i, u_2^i, \dots, u^{c_i}_i\}$. The radius-$t$ neighborhoods of any $u_p^i$ and $u_q^i$ are identical and disjoint: since each node of $U_i$ is chosen as the $2t+1$ or $2t+2$th node in a path of length $4t+4$,
    its radius-$t$ neighborhood cannot see into any other such paths, and because $U_i \subseteq S_{\sC, i} \subseteq S_{\sR, i}$ for all $i \in \{1, 2, \dots, k\}$, each $u_{p}^i$'s radius-$t$ neighborhood is isomorphic to $T_{\gamma}$.
    Since no nodes within the radius-$t$ neighborhood of any $u^i_p$ are labeled when $u^i_p$ is revealed to the algorithm, the algorithm will be unable to distinguish between $u_p^i$ and $u_q^i$ for any $p\ne q$.
    For $i_1\ne i_2$, the radius-$t$ neighborhoods of any two $u_{p_1}^{i_1}$ and $u_{p_2}^{i_2}$ are also disjoint as their nearest common ancestor is a distance of at least $2t$ from at least one of them.
    So every $u_p^i$'s radius-$t$ neighborhood is identical, disjoint, unlabeled at the time of revealing to the algorithm, for all orders of $\seq$. Thus, the algorithm is unable to distinguish between the nodes of $U_i$, and thus, 
    the labeling of $U_i$ is independent of the randomness of the sequence $\seq$, for all $i\in\{1,2,\dotsc,k\}$.
\end{proof}

Note that this also means the orientation---in particular, the half-edge labels which connect a given $u^i_j$ to the rest of the nodes in a path of layer $(\sC, i)$---is indistinguishable to the algorithm.

We now have all necessary lemmas and definitions for the proof of
\cref{thm:unrooted-full}.

\subparagraph*{Proof outline.} We provide a high-level overview of the proof before discussing it with more technicality.
We first prove by induction that all edges in $S_{\sC, i}$ must be labeled with labels in a flexible-SCC of the labels used in $S_{\sR, i}$,
and that all nodes in $S_{\sR, i}$ must be labeled with labels in the trim of the labels used in $S_{\sR, i-1}$. 
This forces the labels of these sets to be contained within a corresponding set of our ``good sequence''---that is, the labels of the nodes in $S_{\sR, i}$ must
belong to $\mathcal{V}_i$, and the labels of $S_{\sC, i}$ must belong to $\mathcal{D}_i$. But this requires that $S_{\sR, k+1}$ be labeled with elements of some $\mathcal{V}_{k+1}$---
which, because the depth of $\Pi$ is $d_{\Pi} = k$ by assumption, is empty. This allows us to conclude that $\algo$ is unable to generate a valid labeling of $G$ with a sufficiently high probability of success.

Throughout the proof, we use the following notation. For each node $v\in G$, we define $\mathcal{V}^{\algo}_{v}$ as the set of all possible node configurations of $v$ that can occur in any run of $\algo$.
Similarly, for any two edges $e_1$ and $e_2$ incident to a node $v$, we define
$\mathcal{D}^{\algo}_{v, e_1, e_2}$ to be the set of all multisets $\{a, b\}\in
\multich{\Sigma}{2}$ such that $\{a, b\}$ is a possible outcome of labeling the
two half-edges $(v, e_1)$ and $(v, e_2)$ when we run $\algo$. 

We now define our induction hypotheses.
\begin{itemize}
  \item \textbf{\boldmath Induction hypothesis for $S_{\sR, i}$.} For each $1\le i\le
  k+1$, any node $v\in S_{\sR, i}$ satisfies $\mathcal{V}^{\algo}_v \subseteq
  \mathcal{V}_i$.
  \item \textbf{\boldmath Induction hypothesis for $S_{\sC, i}$.} For each $1\le i \le k$,
  for each node $v\in S_{\sC, i}$ and any two incident edges $e_1 = (v, u)$ and
  $e_2 = (v, w)$ such that $u$ and $w$ are in layer $(\sC, i)$ or higher, we
  have $\mathcal{D}^{\algo}_{v, e_1, e_2} \subseteq \mathcal{D}_i$.
\end{itemize}

\begin{proof}[Base case: $S_{\sR, 1}$]
    We show that the induction hypothesis for $S_{\sR, 1}$ holds.
    We recall that $\gamma$ was chosen to be the smallest integer such that for each set $S\subseteq \mathcal{V}$ and each $C \in S\setminus \trim(S)$, there exists no correct labeling of $T_{\gamma}^*$ 
    where the node configuration of the root node is $C$ and the node configurations of the remaining $\Delta$-degree nodes are in $S$.
    Now consider any node $v\in S_{\sR, 1}$: by construction, the radius-$\gamma$ neighborhood of $v$ is isomorphic to $T^*_{\gamma}$.
    By setting $S=\mathcal{V}$, we can conclude there is no correct labeling of $G$ such that the labeling chosen for $v$ is in $\mathcal{V}\setminus\trim(\mathcal{V}) = \mathcal{V}\setminus \mathcal{V}_1$. 
    So for any run of $\algo$, the node configuration of $v$ must be contained within $\mathcal{V}_i$---that is, $\mathcal{V}^{\algo}_v\subseteq \mathcal{V}_i$.
\end{proof}

\begin{lemma}[Inductive step: $S_{\sR, i}$]
    For $2\le i\le k$, if the induction hypotheses for $S_{\sR, i-1}$ and $S_{\sC, i-1}$ hold, then the induction hypothesis for $S_{\sR, i}$ holds.
\end{lemma}

\begin{proof}
    Consider any $v\in S_{\sR, i}$. We recall that by the construction of $S_{\sR, i}$, $v$'s radius-$\gamma$ neighborhood is isomorphic to $T^*_{\gamma}$.
    We note that the $\Delta$-degree nodes in this $T^*_{\gamma}$ are the nodes within the radius-$(\gamma - 1)$ neighborhood of $v$. Consider any node $u$ within this neighborhood.
    Since $u\in S_{\sC, i-1}\subseteq S_{\sR, i-1}$, we can apply the inductive hypothesis for $S_{\sR, i-1}$ to conclude that $\mathcal{V}_u^{\algo}\subseteq \mathcal{V}_{i-1}$. 
    Further, since all neighbors of $u$ are in $ S_{\sC, i-1}$, we can apply the
    inductive hypothesis on all edges $e_1, e_2$ incident to $u$, and conclude
    that $\mathcal{D}^{\algo}_{u, e_1, e_2} \subseteq \mathcal{D}_i$.
    From this, it is clear that $\mathcal{V}_u^{\algo}\subseteq \mathcal{S}$,
    where $\mathcal{S}$ consists of node configurations $V\in\mathcal{V}_{i-1}$ such that for all $\{a, b\}\in \multich{\Sigma}{2}$ such that $\{a, b\} \subseteq V$, $\{a, b\}\in \mathcal{D}_i$.
    (Intuitively, this is the restriction of $\mathcal{V_{i-1}}$ to node configurations which have every pair of edge labels $\{a,b\}\in V$ contained in $\mathcal{D}_i$.)

    So all $\Delta$-degree nodes in the $T_{\gamma}^*$ rooted at $v$ have labels in $\mathcal{S}$. So by the definition of $\gamma$, no valid labeling of this $T^*_{\gamma}$ can assign $v$ a node configuration from $\mathcal{S}\setminus \trim(\mathcal{S})$.
    Thus, $\mathcal{V}^{\algo}_v \subseteq \trim(\mathcal{S}) = \mathcal{V}_i$.
\end{proof}

We now prove the inductive hypothesis for $\mathcal{D}_i$. Note that this step is significantly more involved.

\begin{lemma}[Inductive step: $S_{\sC, i}$]
    For $1\le i\le k$, if the induction hypothesis holds for $S_{\sR, i}$, then the induction hypothesis for $S_{\sC, i}$ holds.
\end{lemma}

\begin{proof}
    We first show that all $u_1, u_2\in U_i$ must belong to the same path-flexible strongly connected component. Then, we show that all $v\in S_{\sC, i}$ must also belong to this $\flexSCC$. Finally, we argue that this $\flexSCC$ is indeed $\mathcal{D}_i$.

    Let $u_1, u_2 \in U_i$. Recalling that each $u_i$'s distance from layer $(\sR, i+1)$ was chosen uniformly at random from the set $\{2t+1, 2t+2\}$, let $\alpha$ denote the distance between $u_1$ and $u_2$ if $2t+1$ was the chosen distance from layer $(\sR, i+1)$ for both $u_1$ and $u_2$.
    So the actual distance between $u_1$ and $u_2$ is either $\alpha$ (with probability $1/4$), $\alpha+1$ (with probability $1/2$), or $\alpha+2$ (with probability $1/4$). 
    Further, by \cref{independence-lemma-unrooted}, at the time they are revealed to the algorithm $\algo$, it does not know which of these 3 values is the true distance 
    between $u_1$ and $u_2$, nor does it know which half-edge of each node is contained in the path between $u_1$ and $u_2$. 

    Let $\mathcal{D}_{u_1}$ denote the set of all elements of
    $\multich{\Sigma}{2}$ that can be assigned to 2 half-edges of $u_1$ in some
    fixed correct run of $\algo$. 
    Define $\mathcal{D}_{u_2}$ analogously for $u_2$.
    Since the path between $u_1$ and $u_2$ can be of length $\alpha$,
    $\alpha+1$, or $\alpha+2$, and this length is unknown to the algorithm, any
    $\{ a_1, b_1 \} \in \mathcal{D}_{u_1}$ and $\{ a_2, b_2 \} \in
    \mathcal{D}_{u_2}$ must have walks of all of these lengths in
    $\mathcal{M}_{\mathcal{V}_i}$.
    Since $\alpha+1$ and $\alpha+2$ are clearly coprime, there exists some $N\in\mathbb{N}$ (the Frobenius number of $\alpha+1$ and $\alpha+2$) such that for all $n>N$, there exist $x, y\in \mathbb{N}$ such that $n = x(\alpha+1) + y(\alpha+2)$.
    By definition, then, all such $\{ a_1, b_1 \}$ and $\{ a_2, b_2 \}$ must
    belong to the same path-flexible strongly-connected component of
    $\mathcal{V}_i$. 
    We call this component $\mathcal{D}_U$.
    
    Now, consider any $v\in S_{\sC,i}$, and any two of its incident edges $e_1 =
    \{ v, w_1 \}$ and $e_2 = \{ v, w_2 \}$ such that $w_1, w_2$ are in layer
    $(\sC, i)$ or higher.
    Note that $w_1, w_2$ are therefore in $S_{\sC, i}$ as well. By construction of $S_{\sC, i}$, there is a path $P_1 = (u_1, \dots, w_1, \dots, u_2)$ and a path $P_2 = (u_3, \dots, w_2, \dots, u_4)$ for some $u_1, u_2, u_3, u_4 \in U_i$ (with $u_1 \ne u_2, u_3\ne u_4$). 
    Without loss of generality, suppose $u_1\ne u_4$. Then we can construct a path $P' = (u_1, \dots, w_1, v, w_2, \dots, u_4)$ between $u_1$ and $u_4$ which passes through $e_1$ and $e_2$.
    Let $a_1$ be the edge adjacent to $u_1$ which is traversed in $P'$, and
    $b_1$ any other edge adjacent to $u_1$. Define $a_4$ and $b_4$ analogously
    for $u_4$.

    Let $L(v,e)$ denote the label assigned to the half-edge $(v,e)$.
    By the previous observations, $L(a_1, b_1)$ and $L(a_4, b_4)$ belong to the
    same flexible strongly-connected component, so $L(a_1,b_1)\rightarrow \dots
    \rightarrow L(e_1, e_2) \rightarrow \dots \rightarrow L(a_4, b_4)$
    is a walk in the automaton $\mathcal{M}$ which starts and ends in the same
    flexible strongly connected component.
    Thus, $L(e_1, e_2)$ must also belong to this component---that is,
    $\mathcal{D}_{v, e_1, e_2}\subseteq \mathcal{D}_U$.
    
    Finally, it can be seen that $\mathcal{D}_U$ is indeed $\mathcal{D}_i$ by noting that we must be able to label $S_{\sC, i}$ for $i \in \{1, \dots, k\}$. We have seen that $\mathcal{D}_U$ must be a path-flexible strongly connected component of $\mathcal{V}_i$.
    Thus, if $\mathcal{V}_i$ is a component of a good sequence of length less than $k$, then $\trim(\flexSCC(\mathcal{V}_{k-1}))=\emptyset$ by definition, and we would thus have that $\mathcal{D}_U = \flexSCC(\mathcal{V}_{i-1}) = \emptyset$, and would not be able to label $S_{\sC, k}$. 
\end{proof}

We now have sufficient tools to prove \cref{thm:unrooted-full}:

\begin{proof}[Proof of \cref{thm:unrooted-full}]
    Given a \rolcl algorithm $\algo$ with depth $k = d_\Pi$ which runs with $o(n^{1/k})$ locality, we can derive a contradiction. 
    Taking $G$ to be defined as above, with subsets $S_{\sR, 1}, S_{\sC, 1}, \dots, S_{\sC, k}, S_{\sR, k+1}$, the previous induction argument tells us that $S_{\sR, k+1}$
    must be labeled with elements of some $\mathcal{V}_{k+1} = \trim(\flexSCC(\mathcal{V}_k))$. But if $\mathcal{V}_{k+1}\ne \emptyset$, then we have a good sequence of length $k+1$---a contradiction to the assumption that $d_{\Pi}=k$.
    Thus, no such $o(n^{1/k})$ algorithm can exist, and $\Pi$ has locality $\Omega(n^{1/k})$.
\end{proof}

\section{\boldmath Extending the Gap Between \texorpdfstring{$o(n)$}{o(n)} and \texorpdfstring{$\omega(\sqrt{n})$}{sqrt(n)} in Unrooted Trees}\label{sec:general-gap}

Here, we are concerned with the following result:

\restateThmUnrootedSqrtn*

Our proof is based on the work of Balliu et al.~\cite{almostGlobalLOCAL}.
In said paper the authors show that, in the \local model, any LCL with sublinear
locality $o(n)$ can be transformed into an LCL with locality $O(\sqrt{n})$. 

\subsection{Previous Work}

We provide a brief overview of the proof described in \cite{almostGlobalLOCAL}.

Suppose we have an LCL $\Pi$ which is solved by an algorithm $\algo$ which is known to run with sublinear ($o(n)$) locality.
For any graph $\mathcal{G}$, we can distributedly construct a \enquote{virtual graph} $\mathcal{G}'$ with similar properties to $\mathcal{G}$
by viewing a radius-$O(\sqrt{n})$ neighborhood. $\mathcal{G}'$ is constructed to
be much bigger than $\mathcal{G}$, and has its size upper-bounded by $N\gg n$.
It will be possible to simulate the execution of $\algo$ on $\mathcal{G}'$ by viewing only its corresponding radius-$O(\sqrt{n})$ neighborhoods
in $\mathcal{G}$. 

Given a tree $\tree$, we will now describe how we distributively construct
$\tree'$ using radius $O(\sqrt{n})$ neighborhoods of $\tree$.

Let $\tau = c\sqrt{n}$ for a constant $c$ which depends only on the LCL problem $\Pi$ (we will determine its precise value later).
Given a node $v\in\tree$, we consider its radius-$2\tau$ neighborhood. We construct the \emph{skeleton tree} $\skeltree$ by performing the following operation $\tau$ times:
If a node within the radius-$2\tau$ neighborhood of $v$ is a leaf node, remove
it. Repeating this operation $\tau$ times results in removing all subtrees of
height $\tau$ or less.

Now from $\skeltree$ we construct a set of paths $\pathset_{\tree}$ by first removing all nodes of degree greater than 2. 
We are left with a set of paths which we will call $\skeltree'$. We find a $(c+1, c)$ ruling set $\mathcal{R}$ for these paths. Note that this can be done with constant locality in the \onlinelocal model.
Remove all nodes from $\mathcal{R}$ from $\skeltree'$ to obtain
$\pathset_{\tree}$.

So we have a set $\pathset_{\tree}$ of paths of length $l$ such that $c \le l
\le 2c$. 
From here, we construct~$\tree'$. 

We now define $c = \lpump + 4r$. 
Since $c > \lpump$, we can apply the pumping lemma to make each path in $\pathset$'s corresponding path in $\tree$ arbitrarily long.

\subsection{Constructing a Virtual Tree}
\label{sec:constructing-virtual-tree}

At a high level, we use the following procedure, nearly identical to that described by Balliu et al.~\cite{almostGlobalLOCAL}; the steps are visualized in \cref{fig:pumping-paths}.
For each node we are given, we construct a portion of a virtual graph with many more nodes than our original graph.
This graph consists of radius-$O(\sqrt{n})$ neighborhoods connected by very long paths, which cannot be seen in full by the algorithm.
We show we can run the algorithm on this graph (that is, without knowledge of the parts of the original graph which are outside of the selected node's radius-$O(\sqrt{n})$ neighborhood),
and that the results of this algorithm will give us a valid labeling of the original graph.

\begin{figure}[t]
    \centering
    (a)\includegraphics[scale=0.3,valign=c]{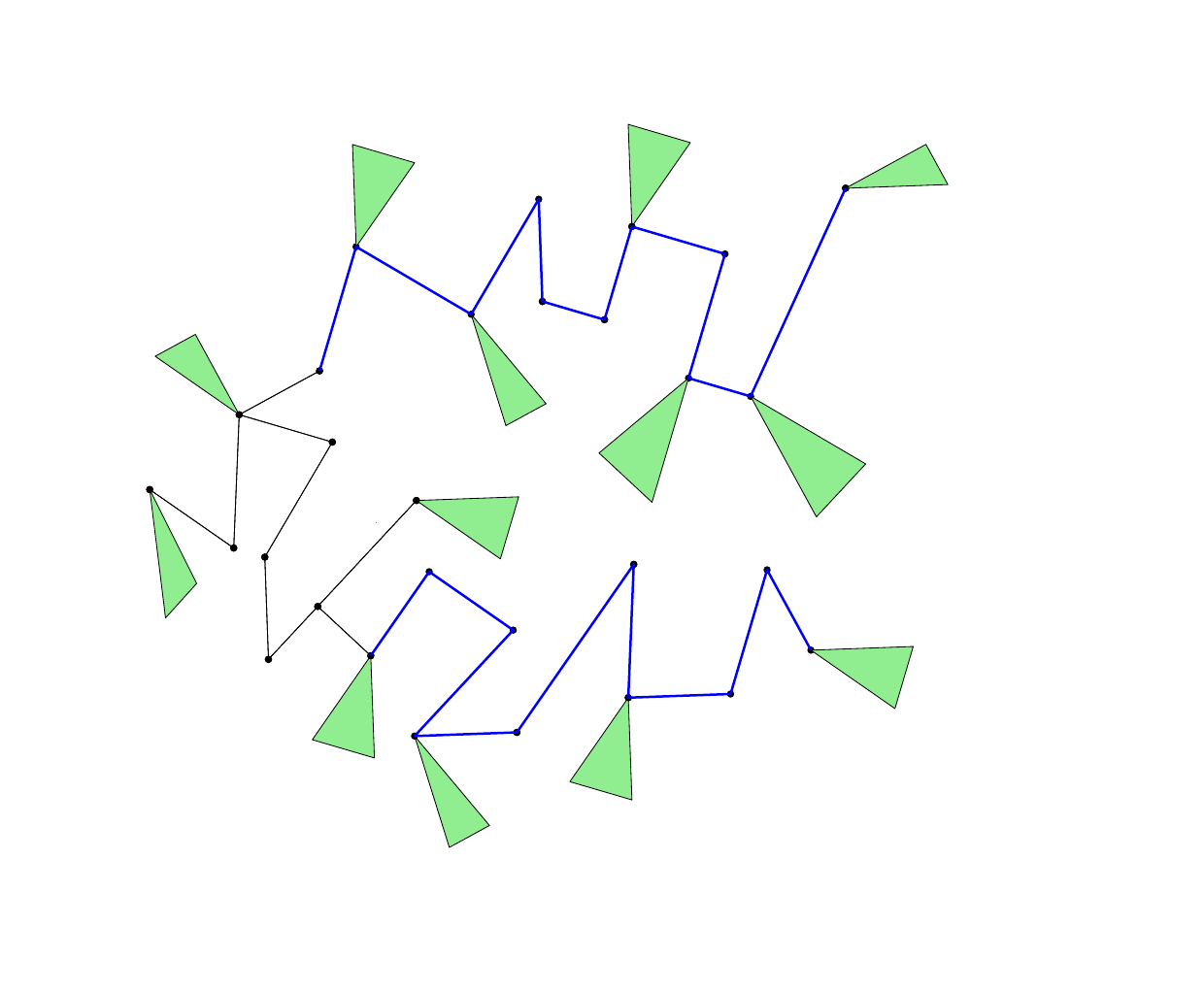}
    (b)\includegraphics[scale=0.3,valign=c]{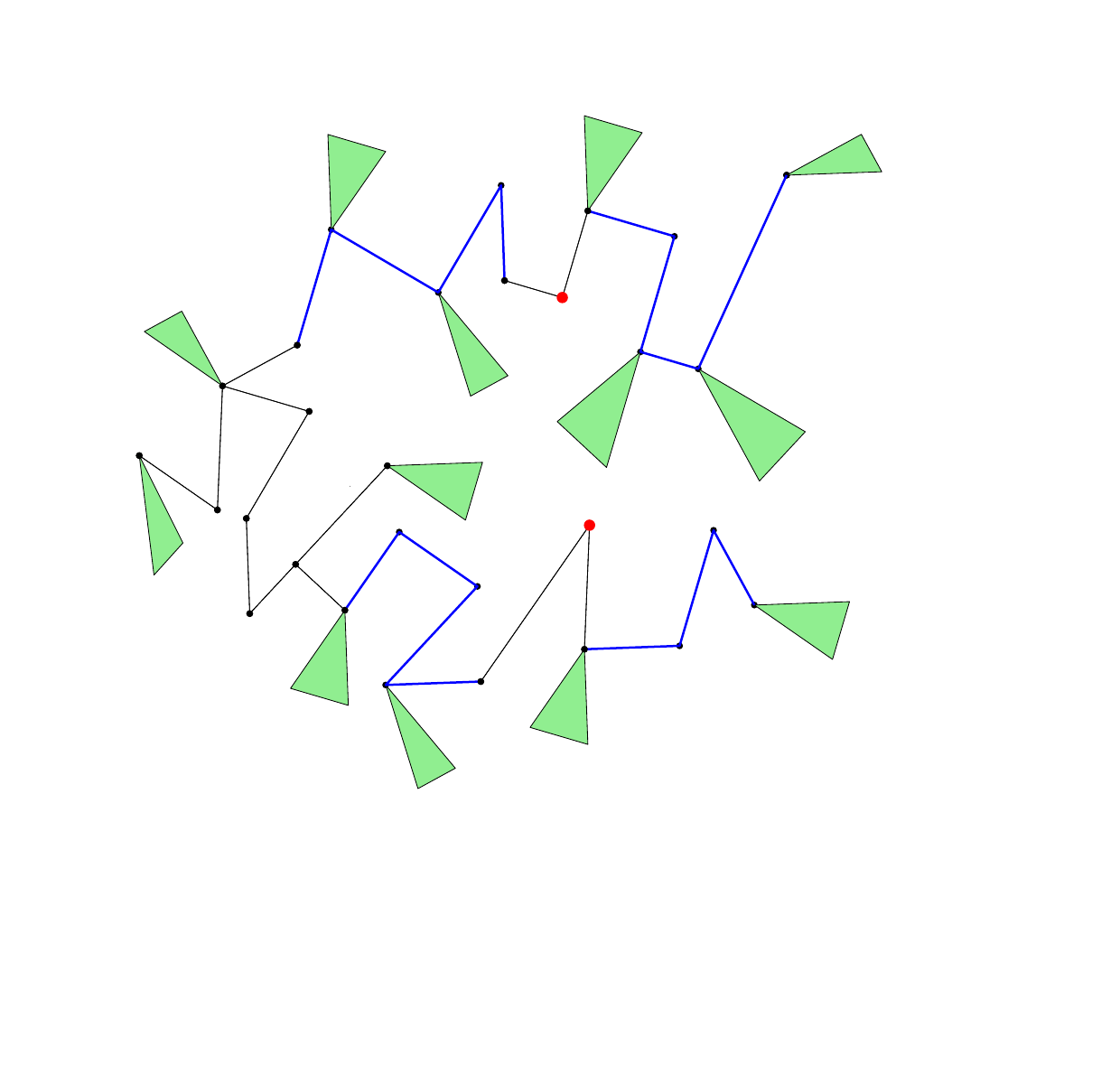}\\
    (c)\includegraphics[scale=0.4,valign=c,trim={0 3cm 0 1cm}]{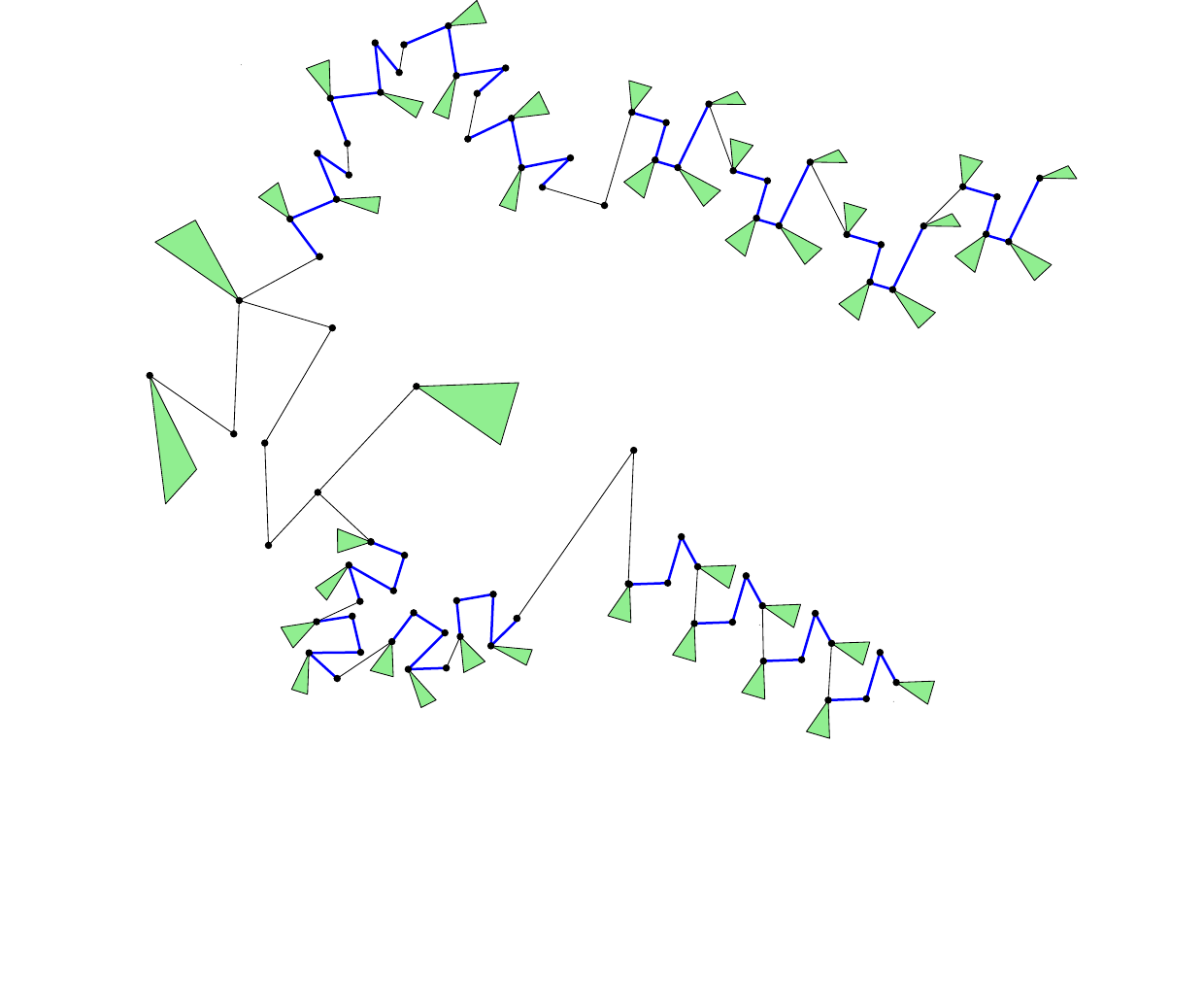}
    \caption{(a)~$T$ with height-$\tau$ trees (green) and long paths (blue) identified.
    (b)~$T$ with ruling set nodes (red) and paths to pump (blue)
    (c)~$S$, with pumped paths.}
    \label{fig:pumping-paths}
\end{figure}

We begin by defining our new virtual graph. Suppose we are given a tree $T=(V,E)$ and an \onlinelocal algorithm $\algo$ with locality $o(n)$.

\subsubsection{Skeleton Tree}

Our first step is to \enquote{prune} $T$, by removing any subtrees of height less than or equal to $\tau = c\sqrt{n}$ (for a constant $c$ which we will define later).
We call this the \emph{skeleton tree} $T'$ of $T$.

More formally, initialize $T'$ as a copy of $T$. For $\tau$ rounds, we remove all leaf nodes from $T'$. After this, $T'$ is our skeleton tree.

Note that we perform this construction distributedly and sequentially: given a node $v$, we check its radius-$2\tau$ neighborhood, and prune the leaf nodes for $\tau$ rounds.
This ensures that every node within a radius-$\tau$ neighborhood will be completely pruned. 

For the remaining nodes in $T'$ (those which were not pruned), we define a function $\psi\colon V(T')\rightarrow V(T)$ which maps nodes of $T'$ to their corresponding node in $T$. 

\subsubsection{A Set of Long Paths}

We now want to construct a set of paths of $T'$ which are of length $l\in[c, 2c]$ for some constant~$c$.

We begin by removing any node with degree greater than $2$ from the skeleton tree $T'$. 
From here, we are left with a set of long paths (note we only worry about those within a radius-$\tau$ neighborhood of whichever node we have selected at a time).
We construct a $(c+1, c)$ ruling set for the remaining nodes.
(Recall that an $(\alpha, \beta)$ ruling set is a set of nodes which are a distance at least $\alpha$ apart, with the property that every other node in the graph is no more than distance $\beta$ from at least one ruling set node.)

Removing all the ruling set nodes, we obtain a set of paths with length at most $2c$.
We discard any paths shorter than $c$ (which occur at the ends of the original long paths we obtain after removing the high-degree nodes). 
We now have a set of paths with length $l \in [c, 2c]$, which we call $\pathset$.

\subsubsection{Virtual Tree}

Before we can finish our definition of the final virtual tree $S$, we need to introduce some additional concepts.
We do this in the next section, then resume the construction of $S$ in
\cref{sec:finishing-virtual-tree}.

\subsection{An Equivalence Relation on Paths and the Pumping Lemma}

We take an \enquote{intermission} to discuss properties of the paths of $\pathset$. 

Given a graph $G$ with a subgraph $H$, we define the \emph{poles} of $H$ to be nodes $v\in V(H)$ which are adjacent to some node in $V(G)\setminus V(H)$.
Now we can define the following equivalence relation on trees (introduced originally in \cite{TimeHierarchyLOCAL}):

\begin{definition}[Equivalence relation $\overset{*}{\sim}$]
    Given a graph $H$ and its poles $F$, define $\xi(H, F) = (D_1, D_2, D_3)$ to
    be a tripartition of $V(H)$ where:
    \begin{itemize}
        \item $D_1 = \bigcup_{v\in F}N^{r-1}(v)$
        \item $D_2 = \bigcup_{v\in D_1} N^{r}(v)\backslash D_1$
        \item $D_3 = V(H) - (D_1\cup D_2) $
    \end{itemize}
    Let $Q$ and $Q'$ be the subgraphs of $H$ and $H'$ induced by the vertices $D_1\cup D_2$ and $D_1'\cup D_2'$ respectively.
    The equivalence holds, i.e., $(H, F)\overset{*}{\sim}(H',F')$, if and only if there is a 1 to 1 correspondence $\phi : (D_1\cup D_2) \rightarrow (D_1'\cup D_2')$ satisfying:
    \begin{itemize}
        \item $Q$ and $Q'$ are isomorphic under $\phi$, preserving the input labels of the LCL problem (if any) and the order of the poles.
        \item Let $\mathcal{L}_*$ be any assignment of the output labels to vertices in $D_1\cup D_2$, and let $\mathcal{L}_*'$ be the corresponding labeling of $D_1'\cup D_2'$ under $\phi$.
        Then $\mathcal{L}_*$ is extendible to $V(H)$ if and only if $\mathcal{L}_*'$ is extendible to $V(H')$. 
    \end{itemize}
\end{definition}

It is proven by Chang and Pettie \cite{TimeHierarchyLOCAL} that if 
the number of poles is constant, there is a constant number of
equivalence classes under $\overset{*}{\sim}$.
We also inherit the following lemma from Chang and Pettie \cite{TimeHierarchyLOCAL}, where $\replace(G, (H, F), (H', F'))$ is the graph obtained by replacing a subgraph $H$ of $G$ with poles $F$ with a new graph $H'$ with poles $F'$.

\begin{lemma}
    Let $G' = \replace(G,(H,F),(H',F'))$. Suppose $(H,F) \overset{*}{\sim} (H',F')$. Let $D_0 = V(G) \setminus V(H)$. Let $\mathcal{L}_{\diamond}$ be a complete labelling of $G$ that is locally consistent for all vertices in $D_2 \cup D_3$. Then there exists a complete labelling $\mathcal{L}'_{\diamond}$ satisfying the following:
	\begin{itemize}
		\item $\mathcal{L}_{\diamond}  = \mathcal{L}'_{\diamond}$ for all $v \in D_0 \cup D_1 \cup D_2$ and their corresponding vertices in $D'_0 \cup D'_1 \cup D'_2$. Also, if $\mathcal{L}_{\diamond}$ is locally consistent for a node $v$, then $\mathcal{L}'_{\diamond}$ is locally consistent for $\phi(v)$.
		\item $\mathcal{L}'_{\diamond}$ is locally consistent for all nodes in $D'_2 \cup D'_3$.
	\end{itemize}
\end{lemma}

We now inherit the notation introduced by Balliu et al.~\cite{almostGlobalLOCAL}: given a tree rooted at $v$, which we denote $\deltree_v$, let $\Class(\deltree_v)$ denote the equivalence class of 
$\deltree_v$ with $v$ as its unique pole. 
Given a path $Q\in\pathset$ with length $k$, we can consider $Q$'s image in the original tree $T$ as a sequence of trees, $(\deltree_i)_{i\in [k]}$, since each node $v\in Q$ is the root of some (possibly empty, aside from $v$ itself) tree which was deleted in the first stage of processing. 
Let $\Type(Q)$ denote the equivalence class of the path $Q$ with its 2 endpoints as poles.

We now restate 2 lemmas:

\begin{lemma}
    Each node $u$ can determine the type of $\deltree_v$ for all $v$ contained within $u$'s radius-$\tau$ neighborhood.
\end{lemma}

This lemma allows us to ignore the deleted subtrees of $T\setminus T'$, as they can be filled in at the end.

\begin{lemma}
    Let $H = (\deltree_i)_{i\in [k]}$ and let $H' = (\deltree_i)_{i\in [k+1]}$ be identical to $H$ in its first $k$ trees. 
    Then $\Type(H')$ is a function of $\Type(H)$ and $\Class(\deltree_{k+1})$.
\end{lemma}

\subsubsection{The Pumping Lemma}

We now have the necessary pieces to define a crucial lemma: the pumping lemma for paths.

\begin{lemma}[Pumping lemma for paths]\label{pumping-lem}
    Let $H = (\deltree_i)_{i\in [k]}$ be a chain of trees with $k\ge \lpump$. Then $H$ can be decomposed into three subpaths $H = x\circ y\circ z$ such that:
    \begin{itemize}
        \item $\abs{xy}\le \lpump$,
        \item $\abs{y}\ge 1$,
        \item $\Type(x\circ y^j \circ z) = \Type(x\circ y \circ z)$ for all
        $j\in\N$.
    \end{itemize}
\end{lemma}

Now, the pumping lemma tells us that in every path $Q\in \pathset$, there is some subpath which can be \enquote{pumped}---repeated an arbitrary number of times---and leave the type of the graph unchanged.
This means that a valid solution of the graph with pumped paths, $S$, can be mapped to the original graph and be used to produce a valid solution.

\subsection{Properties of the Virtual Tree}

We first finish our construction of the virtual tree, which we began in
\cref{sec:constructing-virtual-tree}.
Then we examine some important properties of $S$ which will allow us to simulate $\algo$ on $\tree$ by knowing only a radius-$O(\sqrt{n})$ neighborhood of $T$.

\subsubsection{Finishing the Virtual Tree Construction}
\label{sec:finishing-virtual-tree}

We now are able to finish the construction of the virtual tree $S$.

We choose a parameter $B$ which can be an arbitrarily large function of $n$ (which we will define more precisely later).
At a high level, we will duplicate the pumpable subpath of each path $Q\in\pathset$ so that its new length $l'$ satisfies the inequality: $cB \le l'\le c(B+1)$.
We maintain a mapping of nodes of $T$ which are not a part of any pumped subpath of some $Q\in\pathset$ to their new nodes in $S$:
Let $S_o\subset V(S)$ denote the set of nodes of $S$ which are not part of any
path in $\pathset$, or are at distance at most $2r$ from a node not contained in
any $Q\in\pathset$ (recall $r$ is the checkability radius of our LCL).
Informally, $S_o$ is the set of \enquote{real} nodes of $S$---they are not part of the pumped paths, and they have corresponding nodes in $T$.
Let $\eta : S_o\rightarrow T$ be a mapping of nodes in $S_o$ to their corresponding nodes in $T$.
We also define $T_o = \{\eta(v) \mid v\in S_o\}$.
So $T_o$ is the set of nodes far enough from the pumped regions of $T$ which were not removed in the construction of $T'$.

So we have a virtual tree $S$ which is a function of $T$ and two parameters $c$ and $B$ (which will be defined more precisely soon).

\subsubsection{\boldmath Properties of \texorpdfstring{$S$}{S}}

We state some properties of $S$, proved originally by Balliu et al.~\cite{almostGlobalLOCAL}:

\begin{lemma}
    $S$ has at most $N = c(B+1)n$ nodes, where $n = \abs{V(T)}$.
\end{lemma}

This is because each node in a pumped path or a subtree of a node in a pumped path is duplicated at most $c(B+1)$ times.
Note that $N$ is an \emph{upper} bound on the number of nodes of $S$, but the actual number of nodes $\abs{V(S)}$ will be much smaller.

\begin{lemma}
    For any path $P = (x_1, \dots, x_k)$ of length $k \ge c\sqrt{n}$ which is a subgraph of $T'$, at most $\sqrt{n}/c$ nodes in $V(P)$
    have degree greater than $2$.
\end{lemma}

This lemma gives us an upper bound on the number of high-degree nodes, and is
the point where this procedure demands $\Omega(\sqrt{n})$ locality,
rather than allowing us to obtain an even better locality such as $O(\sqrt[3]{n})$.

Now, we can prove a new lemma which is helpful in the \onlinelocal model.
This tells us no matter what path we take from a given node in $S$, we will either hit a pumped path or a leaf node within a radius-$c\sqrt{n}$ neighborhood.

\begin{corollary}
    Any path $P = (x_1, \dots, x_k) \subset T'$ with length $k \ge c\sqrt{n}$ will contain a subpath $Q\subset\pathset$. 
\end{corollary}
\begin{proof}
    By the previous lemma, $P$ can have at most $\sqrt{n}/c$ nodes which have degree greater than $2$ in $T'$. 
    By the pigeonhole principle, there must be at least one subpath $(v_i, \dots, v_{i+l})$ of $P$ between 2 high-degree nodes (or the endpoints of $P$) with length $l \ge c^2-1>c$ 
    (this will always be true by our choice of $c$), so after selecting ruling set nodes from this subpath, we will be left with at least one path with length between $c$ and $2c$---that is, a path $Q\subset\pathset$. 
\end{proof}
This corollary tells us that any path in $T'$ starting from some vertex $v\in T'$ will either be short enough that it is contained within $v$'s $\tau$-radius neighborhood, or that it contains some subpath which will be pumped.

We now have all necessary ingredients to construct a new algorithm, $\algo'$, which can find a labeling for $T$ with $O(\sqrt{n})$ locality.

\begin{lemma}
    Let $v\in T'$. For any path in $T$ starting at $\psi(v)$ and ending at a leaf node, one of the following holds:
    \begin{itemize}
        \item The path has length less than or equal to $2\tau$.
        \item The path contains some path $Q\in\pathset$ within its first $\tau$ nodes.
    \end{itemize}
    Furthermore, this can be determined with knowledge of only the radius-$2\tau$ neighborhood of $v$.
\end{lemma}

\begin{proof}
    Let $P = (v = v_1, v_2, \dots, v_l)$ be a path in $T$, where $\psi^{-1}(v)$ is defined, and where $v_l$ is some leaf node which is reachable from $v$.
    $l$ denotes the length of the path.
    If $l\le 2\tau$, then we are done. So suppose $l > 2\tau$. We consider $P' = (v=v_1, v_2, \dots, v_{\tau})$---the first $\tau$ nodes of $P$.
    We note that even if we have only seen $v$ at this point in our algorithm,
    we can be certain of the degree
    of the first $\tau$ nodes of $P$ after trimming their height-$\tau$
    subtrees (that is, in $T'$),
    since we can see the $2\tau$-radius neighborhood of $v$ (and thus the $\tau$-radius neighborhood of all nodes in $P'$).
    So $\psi^{-1}(P')\subset T'$. 
    $P'$ has length $\tau$, so we can apply the previous lemma to determine that at most $\frac{\sqrt{n}}{c}$ nodes of $P'$ have degree greater than $2$.
    Define $P''\subset P'$ as the set of paths obtained by removing all nodes with degree greater than $2$ from $P'$.
    By the pigeonhole principle, at least one path in $P''$ must have at least
    $c^2-1$ nodes. $c^2-1 > 2c$ (this will always be true by our choice of $c$). 
    This path therefore must contain a path between 2 ruling set nodes with length between $c$ and $2c$---that is, one of the paths of $\pathset$.
\end{proof}

\subsection{\boldmath Speeding up \texorpdfstring{$\algo$}{A}}

Given a node $v\in T'$ (a node which will not be removed in the construction of the skeleton tree $T'$), we can construct a local portion of $S$ by viewing only $v$'s radius-$2\tau$ neighborhood (which we will call $T_v$).
If every node in this neighborhood is removed in the construction of $T'$, then the connected component containing $v$ has an $O(\sqrt{n})$ diameter and the problem can trivially be solved with $O(\sqrt{n})$ locality.
So suppose $T'_v$ is nonempty.

Let $\ogruntime(n)$ denote the original runtime of $\algo$ on a graph of $n$ nodes. 
The following lemma, reproduced from work by Balliu et al.~\cite{almostGlobalLOCAL}, is essential in determining the value of $B$:

\begin{lemma}
    There exists some constant $c$ such that, if nodes $u,v\in T_o$ are at distance at least $c\sqrt{n}$ in $T$, then their corresponding nodes $\eta^{-1}(u)$ and $\eta^{-1}(v)$ are at distance at least $cB\sqrt{n}/3$ in $S$. 
\end{lemma}

In particular, let $c = 4r + \lpump$.

So we choose $B$ such that $\ogruntime (N)\le cB\sqrt{n}/6$. Such a $B$ will
always exist, as $\ogruntime(x)=o(x)$. 
(Recall that $B$ may be an arbitrarily large function of $n$.)
This choice of $B$, by the previous lemma, implies that if 2 nodes $u,v\in T_o$
are a distance at least $c\sqrt{n}$ in $T$, their radius-$\ogruntime(N)$
neighborhoods are entirely disjoint, and thus one can have no influence on the
other's labeling.

We run $\algo$ on $\eta^{-1}(v)$, but tell $\algo$ that the size of $S$ is (exactly) $N$. Note that since $N$ is an upper bound, it is not possible that $\algo$ sees more than $N$ nodes.

Finally, we can prove the following result:

\begin{lemma}
    For nodes in $T_o$, it is possible to execute $\algo$ on $S$ by knowing only the neighborhood of radius $2c\sqrt{n}$ in $T$.
\end{lemma}

\begin{proof}
    Since $B$ satisfies $\ogruntime (N)\le cB\sqrt{n}/6$ ($\le cB\sqrt{n}/3$), and since by the previous lemma, nodes outside of a radius-$2c\sqrt{n}$ ball in $T_o$ are at distance at least $cB\sqrt{n}/3$ in $S$,
    when $\algo$ runs on $S$ and is processing $v$, it
    cannot see any nodes $u$ such that the distance between $\eta(u)$ and $\eta(v)$ is greater than $2c\sqrt{n}$.
    So the locality of $\algo$ on $S$ is less than the radius of the subtree of $S$ which $\eta(v)$ computed.
    This also implies that the nodes in $T_0$ do not see the whole graph and thus $\algo$ cannot notice that $N$ is not the actual size of the graph $S$.
    Thus, $\algo$ can execute on $S$ by knowing only the radius-$2c\sqrt{n}$ neighborhood of a given node in $T$.
\end{proof}

\subsection{\boldmath LOCALizing \texorpdfstring{$\algo$}{A}}

We now show that $\algo$ can be transformed not only to an \onlinelocal $O(\sqrt{n})$ algorithm, but a \local one.

We note that the procedure followed in constructing $S$ does not rely on any properties unique to the \onlinelocal model--and, in fact, can be done in the \local model, as is shown by Balliu et al.~\cite{almostGlobalLOCAL}.
So we can inherit this construction in the \local model; it only remains to eliminate potential dependencies on global memory and sequential processing.
We do this using a procedure similar to that used by Akbari et al.~\cite{akbari24_online_arxiv}, where we create an \enquote{amnesiac} algorithm $\algo'$. 
This algorithm effectively will see so many nodes and paths of the same type that global memory is of no use, and can be effectively disregarded.

\subsubsection{Preprocessing}

The preprocessing phase involves enumerating all possible pumpable subpaths of the paths in $\pathset$, and feeding them to $\algo$ many times. 
Recall that every such path is a chain of trees of height $\le c\sqrt{n}$
The pumpable portions of a path are bounded by a constant---in particular, they must have length $\le \lpump$.
Further, there is a constant number of equivalence classes (\enquote{Types}) which tree in the path can take on---let us call this number $\xi$.
Then the number of possible pumpable subpaths is bounded by $\xi^{\lpump}$---a constant.

Let $\allpaths$ be a sequence of all possible such paths, in an arbitrary order.
We proceed by feeding each path $P\in\allpaths$ to $\algo$. 
We repeat this many times, until we notice a repeated labeling on each path $P$.
This is guaranteed to happen since there is a constant number of input labels and thus a constant number of possible labelings for each path.
More formally, for any constant $\Delta$, if we feed a given path $P\in \allpaths$ to the original algorithm $\xi^{\lpump}\Delta$ times, by the pigeonhole principle, we are guaranteed to see at least one labeling $\labeling_0^P$ appear $\Delta$ many times.
We call $\labeling_{0}^P$ the \emph{canonical labeling} of $P$.
Finally, we define a function $f : \allpaths \rightarrow \{\labeling_0^P\}_{P\in\allpaths}$, which maps each $P$ to $\labeling_0^P$.

\subsubsection{Running the Algorithm}

From here, we will construct a \local algorithm $\algo'$ which solves $\Pi$ in $O(\sqrt{n})$ rounds.

Given a tree $T$, we first construct $S$ as described above.
This can be done in the \local model with $O(\sqrt{n})$ locality, as proven by Balliu et al.~\cite{almostGlobalLOCAL}.

Let $\pathset_S$ denote the set of all (pumped) paths in $S$.
Each path $Q\in\pathset_S$ takes the form $Q = x\circ y^j \circ z$ for subpaths $x,y,z$ and some constant $j$, with $\abs{y}\le \lpump$,
since we duplicate a middle portion of the original path from $T$ some number $j$ times to obtain $Q$.
$y\in \allpaths$. We find the canonical labeling $\labeling_0^y$ of $y$, and will label each copy of $y\subset Q$ with this labeling.
We then label $x$ and $z$ with brute force. This is doable in constant locality,
since $\abs{x\circ y\circ z}\le 2c$, a constant. 

From here, we can fill in the remaining portions of the graph---those not included in any path in $\pathset$.
By lemma 4.3, the remaining nodes not in any deleted tree or pumped path can be labeled with knowledge of a $\tau$-radius neighborhood.
Similarly, by lemma 3.3, each node in some deleted tree can be labeled with only knowledge of the tree itself.
Since each of these trees has height at most $\tau$, a locality of $\tau$ is sufficient to label each of these nodes using $\algo$. 

So we are able to run $\algo$ with $O(\sqrt{n})$ locality, without reliance on any properties of the \onlinelocal model---and in particular, through use of the \local model.
With this procedure, we can transform any \onlinelocal $o(n)$ algorithm to a \local algorithm which runs with $O(\sqrt{n})$ locality.

Since the paths are bounded by a constant length, this can be done with constant locality.
Further, there is a (likely very large) constant number of paths, so this process will terminate.
We will repeat this many times, and notice that our algorithm will begin to repeat labelings of a given path.
From here, we can paste the labelings to the paths of $\pathset$. 

More formally, we notice that the paths of $\pathset$ have their lengths bounded by a constant, namely $2c$.
Further, the set of input labels $\Sigma_{\iin}$ has constant size.
So there is a constant number of possible paths which can occur in $\pathset$.

\bibliography{paper}

\begin{thebibliography}{10}

\bibitem{akbari24_online_arxiv}
Amirreza Akbari, Xavier Coiteux-Roy, Francesco D'Amore, Fran{\c c}ois {Le
  Gall}, Henrik Lievonen, Darya Melnyk, Augusto Modanese, Shreyas Pai,
  Marc-Olivier Renou, V{\'a}clav Rozhon, and Jukka Suomela.
\newblock Online locality meets distributed quantum computing, 2024.
\newblock \href {https://arxiv.org/abs/2403.01903} {\path{arXiv:2403.01903}}.

\bibitem{akbari_et_al:LIPIcs.ICALP.2023.10}
Amirreza Akbari, Navid Eslami, Henrik Lievonen, Darya Melnyk, Joona
  S{\"a}rkij{\"a}rvi, and Jukka Suomela.
\newblock Locality in online, dynamic, sequential, and distributed graph
  algorithms.
\newblock In Kousha Etessami, Uriel Feige, and Gabriele Puppis, editors, {\em
  50th International Colloquium on Automata, Languages, and Programming,
  {ICALP} 2023, July 10-14, 2023, Paderborn, Germany}, volume 261 of {\em
  LIPIcs}, pages 10:1--10:20. Schloss Dagstuhl - Leibniz-Zentrum f{\"u}r
  Informatik, 2023.
\newblock \href {https://doi.org/10.4230/LIPICS.ICALP.2023.10}
  {\path{doi:10.4230/LIPICS.ICALP.2023.10}}.

\bibitem{balliu2022efficient}
Alkida Balliu, Sebastian Brandt, Yi-Jun Chang, Dennis Olivetti, Jan
  Studen{\'y}, and Jukka Suomela.
\newblock Efficient classification of locally checkable problems in regular
  trees.
\newblock In Christian Scheideler, editor, {\em 36th International Symposium on
  Distributed Computing, {DISC} 2022, October 25-27, 2022, Augusta, Georgia,
  {USA}}, volume 246 of {\em LIPIcs}, pages 8:1--8:19. Schloss Dagstuhl -
  Leibniz-Zentrum f{\"u}r Informatik, 2022.
\newblock \href {https://doi.org/10.4230/LIPICS.DISC.2022.8}
  {\path{doi:10.4230/LIPICS.DISC.2022.8}}.

\bibitem{balliu23_locally_dc}
Alkida Balliu, Sebastian Brandt, Yi-Jun Chang, Dennis Olivetti, Jan
  Studen{\'y}, Jukka Suomela, and Aleksandr Tereshchenko.
\newblock Locally checkable problems in rooted trees.
\newblock {\em Distributed Computing}, 36(3):277--311, 2023.
\newblock \href {https://doi.org/10.1007/S00446-022-00435-9}
  {\path{doi:10.1007/S00446-022-00435-9}}.

\bibitem{almostGlobalLOCAL}
Alkida Balliu, Sebastian Brandt, Dennis Olivetti, and Jukka Suomela.
\newblock Almost global problems in the {LOCAL} model.
\newblock {\em Distributed Computing}, 34(4):259--281, 2021.
\newblock \href {https://doi.org/10.1007/S00446-020-00375-2}
  {\path{doi:10.1007/S00446-020-00375-2}}.

\bibitem{balliu24_shared_arxiv}
Alkida Balliu, Mohsen Ghaffari, Fabian Kuhn, Augusto Modanese, Dennis Olivetti,
  Mika{\"e}l Rabie, Jukka Suomela, and Jara Uitto.
\newblock Shared randomness helps with local distributed problems, 2024.
\newblock \href {https://arxiv.org/abs/2407.05445} {\path{arXiv:2407.05445}}.

\bibitem{balliu23_sinkless_sosa}
Alkida Balliu, Janne~H. Korhonen, Fabian Kuhn, Henrik Lievonen, Dennis
  Olivetti, Shreyas Pai, Ami Paz, Joel Rybicki, Stefan Schmid, Jan Studen{\'y},
  Jukka Suomela, and Jara Uitto.
\newblock Sinkless orientation made simple.
\newblock In Telikepalli Kavitha and Kurt Mehlhorn, editors, {\em 2023
  Symposium on Simplicity in Algorithms, {SOSA} 2023, Florence, Italy, January
  23-25, 2023}, pages 175--191. {SIAM}, 2023.
\newblock \href {https://doi.org/10.1137/1.9781611977585.CH17}
  {\path{doi:10.1137/1.9781611977585.CH17}}.

\bibitem{brandt16_lower_stoc}
Sebastian Brandt, Orr Fischer, Juho Hirvonen, Barbara Keller, Tuomo
  Lempi{\"a}inen, Joel Rybicki, Jukka Suomela, and Jara Uitto.
\newblock A lower bound for the distributed lov{\'a}sz local lemma.
\newblock In Daniel Wichs and Yishay Mansour, editors, {\em Proceedings of the
  48th Annual {ACM} {SIGACT} Symposium on Theory of Computing, {STOC} 2016,
  Cambridge, MA, USA, June 18-21, 2016}, pages 479--488. {ACM}, 2016.
\newblock \href {https://doi.org/10.1145/2897518.2897570}
  {\path{doi:10.1145/2897518.2897570}}.

\bibitem{chang20_complexity_disc}
Yi-Jun Chang.
\newblock The complexity landscape of distributed locally checkable problems on
  trees.
\newblock In Hagit Attiya, editor, {\em 34th International Symposium on
  Distributed Computing, {DISC} 2020, October 12-16, 2020, Virtual Conference},
  volume 179 of {\em LIPIcs}, pages 18:1--18:17. Schloss Dagstuhl -
  Leibniz-Zentrum f{\"u}r Informatik, 2020.
\newblock \href {https://doi.org/10.4230/LIPICS.DISC.2020.18}
  {\path{doi:10.4230/LIPICS.DISC.2020.18}}.

\bibitem{chang19_exponential_siamjc}
Yi-Jun Chang, Tsvi Kopelowitz, and Seth Pettie.
\newblock An exponential separation between randomized and deterministic
  complexity in the {LOCAL} model.
\newblock {\em SIAM Journal on Computing}, 48(1):122--143, 2019.
\newblock \href {https://doi.org/10.1137/17M1117537}
  {\path{doi:10.1137/17M1117537}}.

\bibitem{TimeHierarchyLOCAL}
Yi-Jun Chang and Seth Pettie.
\newblock A time hierarchy theorem for the {LOCAL} model.
\newblock {\em SIAM Journal on Computing}, 48(1):33--69, 2019.
\newblock \href {https://doi.org/10.1137/17M1157957}
  {\path{doi:10.1137/17M1157957}}.

\bibitem{no_distributed_quantum_advantage}
Xavier Coiteux-Roy, Francesco d'Amore, Rishikesh Gajjala, Fabian Kuhn,
  Fran\c{c}ois Le~Gall, Henrik Lievonen, Augusto Modanese, Marc-Olivier Renou,
  Gustav Schmid, and Jukka Suomela.
\newblock No distributed quantum advantage for approximate graph coloring.
\newblock In {\em Proceedings of the 56th Annual ACM Symposium on Theory of
  Computing}, STOC 2024, page 1901–1910, New York, NY, USA, 2024. Association
  for Computing Machinery.
\newblock \href {https://doi.org/10.1145/3618260.3649679}
  {\path{doi:10.1145/3618260.3649679}}.

\bibitem{ghaffari18_derandomizing_focs}
Mohsen Ghaffari, David~G. Harris, and Fabian Kuhn.
\newblock On derandomizing local distributed algorithms.
\newblock In Mikkel Thorup, editor, {\em 59th {IEEE} Annual Symposium on
  Foundations of Computer Science, {FOCS} 2018, Paris, France, October 7-9,
  2018}, pages 662--673. {IEEE} Computer Society, 2018.
\newblock \href {https://doi.org/10.1109/FOCS.2018.00069}
  {\path{doi:10.1109/FOCS.2018.00069}}.

\bibitem{DBLP:conf/stoc/GhaffariKM17}
Mohsen Ghaffari, Fabian Kuhn, and Yannic Maus.
\newblock On the complexity of local distributed graph problems.
\newblock In Hamed Hatami, Pierre McKenzie, and Valerie King, editors, {\em
  Proceedings of the 49th Annual {ACM} {SIGACT} Symposium on Theory of
  Computing, {STOC} 2017, Montreal, QC, Canada, June 19-23, 2017}, pages
  784--797. {ACM}, 2017.
\newblock \href {https://doi.org/10.1145/3055399.3055471}
  {\path{doi:10.1145/3055399.3055471}}.

\bibitem{ghaffari17_distributed_soda}
Mohsen Ghaffari and Hsin-Hao Su.
\newblock Distributed degree splitting, edge coloring, and orientations.
\newblock In Philip~N. Klein, editor, {\em Proceedings of the Twenty-Eighth
  Annual {ACM-SIAM} Symposium on Discrete Algorithms, {SODA} 2017, Barcelona,
  Spain, Hotel Porta Fira, January 16-19}, pages 2505--2523. {SIAM}, 2017.
\newblock \href {https://doi.org/10.1137/1.9781611974782.166}
  {\path{doi:10.1137/1.9781611974782.166}}.

\bibitem{grunau22_landscape_podc}
Christoph Grunau, V{\'a}clav Rozhon, and Sebastian Brandt.
\newblock The landscape of distributed complexities on trees and beyond.
\newblock In Alessia Milani and Philipp Woelfel, editors, {\em {PODC} '22:
  {ACM} Symposium on Principles of Distributed Computing, Salerno, Italy, July
  25 - 29, 2022}, pages 37--47. {ACM}, 2022.
\newblock \href {https://doi.org/10.1145/3519270.3538452}
  {\path{doi:10.1145/3519270.3538452}}.

\bibitem{holroyd16_finitely_fmpi}
Alexander~E. Holroyd and Thomas~M. Liggett.
\newblock Finitely dependent coloring.
\newblock {\em Forum of Mathematics, Pi}, 4, 2016.
\newblock \href {https://doi.org/10.1017/fmp.2016.7}
  {\path{doi:10.1017/fmp.2016.7}}.

\bibitem{linial92_locality_siamjc}
Nathan Linial.
\newblock Locality in distributed graph algorithms.
\newblock {\em SIAM Journal on Computing}, 21(1):193--201, 1992.
\newblock \href {https://doi.org/10.1137/0221015} {\path{doi:10.1137/0221015}}.

\bibitem{naor91_lower_siamjdm}
Moni Naor.
\newblock A lower bound on probabilistic algorithms for distributive ring
  coloring.
\newblock {\em SIAM Journal on Discrete Mathematics}, 4(3):409--412, 1991.
\newblock \href {https://doi.org/10.1137/0404036} {\path{doi:10.1137/0404036}}.

\bibitem{whatcanbecomputedlocally}
Moni Naor and Larry~J. Stockmeyer.
\newblock What can be computed locally?
\newblock {\em SIAM Journal on Computing}, 24(6):1259--1277, 1995.
\newblock \href {https://doi.org/10.1137/S0097539793254571}
  {\path{doi:10.1137/S0097539793254571}}.

\end{thebibliography}

\end{document}